\documentclass[11pt]{article}

\usepackage{amssymb,amsmath,algorithm2e,amsthm}
\usepackage{complexity,tikz}
\usepackage{forest,cite,fullpage}
\usepackage{bm}

\usepackage{amsfonts}
\usepackage{[palatino}
\theoremstyle{definition}
\newtheorem{definition}{Definition}[section]
\theoremstyle{theorem}

\newtheorem{lemma}{Lemma}[section]

\newtheorem{theorem}{Theorem}[section]
\newtheorem{proposition}{Proposition}[section]
\newtheorem{observation}{Observation}[section]

\newtheorem{remark}{Remark}[section]
\usetikzlibrary{arrows.meta}
\tikzset{
	0 my edge/.style={my edge, my edge},
	my edge/.style={-{Stealth[]}},
}

\forestset{
	BDT/.style={
		for tree={
			if n children=0{}{circle},
			draw,
			edge={
				my edge,
			},
			if n=1{
				edge+={0 my edge},
			}{},
			font=\sffamily,
		}
	},
}

\usepackage{color}
\usepackage[pdfstartview=FitH,pdfpagemode=UseNone,colorlinks=true,citecolor=blue,linkcolor=blue, backref=page]{hyperref}
\usepackage{setspace}
\title{Power of Decision Trees with Monotone Queries}

\author{Prashanth Amireddy\thanks{Harvard University. This work was done while the author was an undergraduate student at IIT Madras.} \and Sai Jayasurya\thanks{This work was done while the author was an undergraduate student at IIT Madras.} \and
Jayalal Sarma\thanks{Indian Institute of Technology Madras.}}
\newcommand{\introthm}[2]{\vspace{1mm}\noindent \textbf{Theorem~\ref{#1}.} \textit{#2} \vspace{1mm}}

\usepackage{collect}
\def\movetoappendix{1}
\definecollection{appendix}
\makeatletter
\newenvironment{aproof}[2]
  { \@nameuse{collect}{appendix}
  { \subsection{#1} \label{#2} \begin{proof} } {\end{proof}}
  }{\@nameuse{endcollect}}
\makeatother

\makeatletter
\newenvironment{appsection}[2]
  { \@nameuse{collect}{appendix}
  { \subsection{#1} \label{#2} }
  {}
  }{\@nameuse{endcollect}}
\makeatother

\ifthenelse{\equal{\movetoappendix}{0}}{
        \renewenvironment{aproof}[2]{\begin{proof} \color{gray}  } {\end{proof} }
        
}{}
\newcommand{\bigO}{\mathcal{O}}

\newcommand{\zo}{\{0,1\}}
\newcommand{\zon}{\zo^n}

\renewcommand{\bar}[1]{\overline{\vphantom{1}#1}}

\newcommand{\DT}{\mathsf{DT}}
\newcommand{\DTm}{\mathsf{DT}_m}
\newcommand{\Cm}{\mathsf {C}_m}
\newcommand{\DTmna}{\mathsf{DT}_m^{na}}
\newcommand{\DTmr}{\mathsf{DT}_m^{r}}
\newcommand{\DTmrbag}[1]{{\sf DT}^{R,#1}_m}
\newcommand{\Cmna}{\mathsf{C}_m^{na}}

\newcommand{\DLm}{\mathsf{DL}_m}
\newcommand{\Th}{\mathsf{Th}}
\newcommand{\fdim}{\mathsf{fdim}}

\newcommand{\MDT}[1]{{\mathsf{DT}}(\textit{mon-}#1)}
\newcommand{\MDTT}[1]{{\mathsf{DT}}^{#1}}
\newcommand{\MDTTT}[2]{{\mathsf{DT}}^{#1}(\textit{mon-}#2)}

\newcommand{\RMDT}[1]{{\sf{RDT}}(\textit{mon-}#1)}
\newcommand{\MDL}[1]{{\mathsf{DL}}(\textit{mon-}#1)}

\newcommand{\alt}{\mathsf{alt}}

\newcommand{\calC}{{\cal{C}}}

\begin{document}
\maketitle	

\begin{abstract}
In this paper, we initiate study of the computational power of adaptive and non-adaptive monotone decision trees -- decision trees where each query is a monotone function on the input bits. In the most general setting, the monotone decision tree height (or size) can be viewed as a \textit{measure of non-monotonicity} of a given Boolean function. We also study the restriction of the model by restricting (in terms of circuit complexity) the monotone functions that can be queried at each node. This naturally leads to complexity classes of the form $\MDT{\calC}$ for any circuit complexity class $\calC$, where  the height of the tree is $\bigO(\log n)$, and the query functions can be computed by monotone circuits in the class $\calC$. In the above context, we prove the following characterizations and bounds.
\begin{itemize}
\item For any Boolean function $f$, we show that the minimum monotone decision tree height can be exactly characterized (both in the adaptive and non-adaptive versions of the model) in terms of its {\em alternation} ($\alt(f)$ is defined as the maximum number of times that the function value changes, in any chain in the Boolean lattice). We also characterize the non-adaptive decision tree height with a natural generalization of certification complexity of a function. Similarly, we determine the complexity of non-deterministic and randomized variants of monotone decision trees in terms of $\alt(f)$.
\item We show that $\MDT{\calC} = \calC$ when $\calC$ contains monotone circuits for the threshold functions (for e.g., if $\calC = \TC^0$). For $\calC = \AC^0$, we are able to show that any function in $\AC^0$ can be computed by a sub-linear height monotone decision tree with queries having monotone $\AC^0$ circuits.
\item To understand the logarithmic height case in case of $\AC^0$ i.e., $\MDT{\AC^0}$, we show that for any $f$ (on $n$ bits) in $\MDT{\AC^0}$, and for any positive constant $\epsilon\le 1$, there is an $\AC^0$ circuit for $f$ with $\bigO(n^\epsilon)$ negation gates. 
\end{itemize}
En route our main proofs, we study the monotone variant of the decision list model, and prove corresponding characterizations in terms of $\alt(f)$ and also derive as a consequence that $\MDT{\calC} = \MDL{\calC}$ if $\calC$ has appropriate closure properties (where $\MDL{\calC}$ is defined similar to $\MDT{\calC}$ but for {\em decision lists}).
\end{abstract}

\tableofcontents

\section{Introduction}
\label{sec:intro}
The {\em decision tree} model is a fundamental abstraction that captures computation appearing in various scenarios, ranging from query based decision making procedures to learning algorithms for Boolean functions. The model represents the algorithmic steps in order to compute a Boolean function $f:\{0,1\}^n \to \{0,1\}$, as a sequence of branching operations based on queries to the input bits and the branching depends on the result of the query. It is quite natural to view the branching as a rooted binary tree where the leaves of the tree are labeled with 0 or 1 to represent value of the function if the computation reaches that leaf.

The simplest form studied is when the queries are directly to bits of the input\cite{Juk12,BW02} -- and hence the nodes of a decision tree (except for leaves) are labeled with input variables which it queries. For a Boolean function $f$, the deterministic decision tree complexity, $\DT(f)$, is the minimum height of any decision tree computing $f$. By \emph{height}, we always refer to the maximum number of internal nodes in a path from root to a leaf. The size of the decision tree, which is defined as the number of leaves in the tree is an independently interesting measure of complexity of $f$, and indeed, since the tree is binary, the size cannot be more than exponential in $\DT(f)$. Generalization of the model of decision trees in the algorithmic setting has been studied -- namely randomized and quantum decision trees (see \cite{BW02}). Decision trees can be adaptive and non-adaptive depending on whether, in the algorithm, the next query depends on the Boolean result of the previous queries or not. In the interpretation of the tree, this translates to whether the tree queries the same variable at all nodes in the same level.

The (adaptive) decision tree height, $\DT(f)$ is related to many fundamental complexity measures of Boolean functions. It is known to be polynomially related to degree of $f$ over $\mathbb{R}$, block sensitivity, certificate complexity (see survey \cite{BW02}) and with the recent resolution of sensitivity conjecture \cite{huang2019induced}, even to sensitivity of the Boolean function $f$. Non-adaptive decision trees are not as powerful. 

An important way of generalizing the decision tree model is by allowing stronger queries than the individual bit queries. One of the well-studied models in this direction is that of parity decision trees where each query is a parity of a subset of input bits~\cite{KM93}. Each node in the tree is associated with a subset $S \subseteq [n]$
\footnote{We denote the set $\{1,2,\dots ,n\}$ by $[n]$.} and the query to the input at the node is the function $\oplus_{i \in S} x_i$, where $x_i$ stands for the $i^{th}$ bit of $x$. The model of parity decision trees received a lot of attention due to its connection to a special case of log-rank conjecture known as the XOR-log-rank conjecture \cite{SZ10}. The conjecture, in particular, implies that the non-adaptive $(\DT_{\oplus}^{\sf na}(f))$ and adaptive $(\DT_{\oplus}(f))$ parity decision complexity measures of functions are not polynomially related in general\footnote{If ${\sf supp}(f) = \{S \subseteq [n] \mid \hat{f}(S) \ne 0\}$, ${\sf sps}(f) = |{\sf supp}(f)|$ and ${\sf fdim}(f) = \dim({\sf supp}(f))$, then by\cite{SZ10}, $\log \mathsf{sps}(f)/2 \le \mathsf{DT_{\oplus}}(f) \le$ {\sf fdim}$(f) = \DT_{\oplus}^{\sf na}(f)$\cite{POSSW11,San19}. The XOR-logrank conjecture\cite{SZ10} states that $\mathsf{DT_{\oplus}}(f) \le \poly\left(\log \mathsf{sps}(f)\right)$, and $\exists f$ for which $\fdim(f)$ and $\log({\sf sps}(f))$ are exponentially far apart.}.

Other well-studied generalizations of the standard decision tree model include \textit{linear decision trees} \cite{DL78,Sni81,YR80} (where each node queries a linear function of the form $\sum_i \alpha_i x_i + \beta > 0$)  and \textit{algebraic decision trees} \cite{SY80,BLW92,Ben83} (where each node queries the sign of a polynomial evaluation of degree at most $d$ in terms of the input variable). Polynomial size linear  decision trees can compute knapsack problem which is $\NP$-complete and the above studies prove exponential size lower bounds for explicit languages. Ben-Asher and Newman~\cite{BN95}, studied the decision trees when conjunction and disjunction of variables are allowed as queries on the internal nodes and showed lower bounds for the height of such decision trees required to compute the threshold functions.\\[-3mm]

\noindent{\bf Our results:} 

\noindent We initiate the study of a new generalization of the decision tree model based on allowing more general queries. The most general version of our model allows the algorithm to query arbitrary {\em monotone} functions on the input\footnote{Indeed, this generalized model is still universal since in normal decision trees, the queries are monotone functions.}.  We define the deterministic monotone decision tree complexity of a function $f$, denoted by $\DTm(f)$ to be the minimum height of any decision tree with monotone queries at each node, that computes $f$. When the decision tree is non-adaptive (i.e., when query functions do not depend on the result of previous queries) we denote it by $\DTmna(f)$.\\[-3mm]

\noindent{\bf $\DTm$ and $\DTmna$ as measures of non-monotonicity:} Monotone decision tree complexity measures can also be interpreted as a measure of non-monotonicity of the function $f$. Our first result is an exact characterization of this measure in terms of a well-studied measure of non-monotonicity called {\em alternation}. Our first main result is the following connection between the monotone decision tree height and the alternation of the function in the case of adaptive and non-adaptive setting. They are exponentially far apart similar to what is conjectured in the case of {\em parity decision trees}.

\begin{theorem}\label{thm:intro-all-dtm}
For any Boolean function $f$, $\DTm(f) = \lceil \log(\mathsf{alt}(f) + 1) \rceil$, and $\DTmna(f) = \mathsf{alt}(f)$.
\end{theorem}

En route to proving the above theorem, we also relate a similar generalization of a well-studied computational model called decision lists (see Section~\ref{sec:prelims} for a definition). If $\DLm(f)$ stands for the minimum length of any monotone decision list computing a Boolean function $f$, then we show that, $\DLm(f) = \mathsf{alt}(f)+1$. We also provide a natural generalization of certificate complexity of a Boolean function, denoted by $\Cm$ (and its non-adaptive version denoted by $\Cmna$) and show that for every function $f$, $\Cmna(f)=\DTmna(f)$ (Proposition~\ref{prop:cna-dtna}).\\[-2mm]


\noindent{\bf Non-deterministic and randomized monotone decision trees:}
We study non-deterministic and randomized monotone decision trees (see Sections~\ref{sec:ndt} and~\ref{sec:rdt}) and consider variants of the definitions, and show equivalences and bounds. In particular, we show constant upper bounds for the height of non-deterministic monotone decision trees (Theorem~\ref{prop:nmdt}) and show characterizations for the height of the randomized version in terms of deterministic monotone decision tree complexity, and thus alternation (Theorem~\ref{thm:rmdt-main}). \\[-3mm]

\noindent{\bf Decision trees with  restricted (monotone) queries:}
While the above models provide a measure of non-monotonicity, one of the main drawbacks of the above decision tree model is that, the computational model is not succinctly representable. It is natural to see if we can restrict the query functions to circuit complexity classes which allow succinct representation for functions. An immediate direction is to understand the power of the model if the query functions are restricted to circuit complexity classes; studied in Section~\ref{sec:mdt-query-restrict}.
More formally, we define $\MDT{\calC}$ to be the class of functions that can be computed by monotone decision trees of height $\bigO(\log n)$ where each query function has a monotone circuit in $\calC$, or equivalently all the queries belong to $\textit{mon-}\calC$. 

To justify the bound of $\bigO(\log n)$ on the height of monotone decision trees, we show that if we allow the upper bound on height to be asymptotically different from $\Theta(\log n)$, then the class of functions computed by the model will be different from $\calC$ \footnote{We assume that in$\calC$, all the circuits are polynomial sized, and that there is at least one function with $\Omega(n)$ alternation. This is true for the  complexity classes $\AC^0, \TC^0, \NC^1$ etc.}. More precisely, if $\MDTT{d(n)}$ denotes the class of functions computed by monotone decision trees of height at most $d(n)$ (thus $\MDT{\calC} \equiv \MDTTT{\bigO(\log n)}{\calC}$), we show that, for any $g(n) = o(\log n)$, and $h(n) = \omega(\log n)$, $\MDTTT{g(n)}{\calC} \subsetneq \calC$ and $\MDTTT{h(n)}{
\calC} \nsubseteq \calC$. This justifies the question of $\MDTTT{\bigO(\log n)}{\calC}$ vs $\calC$, which we answer (in some cases) in the following theorem.
\begin{theorem}
\label{introthm:mdtc-c}
For any circuit complexity class $\calC$ such that $\textit{mon-}\TC^0 \subseteq \textit{mon-}\calC$, 
$\MDT{\calC} = \calC$.
\end{theorem}

Hence, in particular, $\MDT{\TC^0} = \TC^0$. The situation when $\calC$ does not contain $\TC^0$ is less understood. We start by arguing that all functions in $\AC^0$ can be computed by monotone decision trees in sub-linear height. More specifically, for any constant $r$, $\AC^0 \subseteq \MDTTT{d(n)}{\AC^0}$ when $d(n) = \Omega\left(\frac{n}{\log^r n}\right)$ (Theorem~\ref{thm:sublinear-ac0}).
It is natural to ask whether the sub-linear height can be improved further.
In particular, whether $\MDT{\AC^0}$ is equal to $\AC^0$ or not.
Towards this, by using a technique from \cite{SW93}, we first show a negation limited circuit for functions in $\MDT{\AC^0}$:

\begin{theorem}
\label{introthm:ne-ac0}
If a Boolean function $f$ on $n$ variables is in $\MDT{\AC^0}$, then for any positive constant $\epsilon\le 1$, there is an $\AC^0$ circuit for $f$ with $\bigO(n^\epsilon)$ negation gates.
\end{theorem}

\begin{remark}
	By Theorem 4.3 of \cite{SW93}, we know that the negation function ${\sf Neg}:\{0,1\}^n \to \{0,1\}^n$ defined as ${\sf Neg}(x)=x \oplus 1^n$ cannot be computed by circuits of constant depth and $\bigO(\sqrt{n})$ negations. We note that this does not immediately show that $\MDT{\AC^0} \ne \AC^0$  as the output of ${\sf Neg}$ is not single-bit.
\end{remark}

In a tight contrast to Theorem~\ref{introthm:ne-ac0}, it can be derived using \cite{SW93} that if $f \in \AC^0$ with $\alt(f) = \Omega(n)$, then any $\AC^0$ circuit computing it must have at least $\Omega(n^\epsilon)$ negations for some constant $\epsilon > 0$ (see Theorem~\ref{thm:neg-ac0-lb}). Thus, an asymptotic improvement to this, with respect to the number of negations, would imply $\MDT{\AC^0} \ne \AC^0$.

En route these main results, we also note that the analogously defined class of functions $\MDL{\calC}$ for decision lists (defined in Section~\ref{sec:prelims}) is exactly equal to $\MDT{\calC}$. Defining $\RMDT{\calC}$ similar to $\MDT{\calC}$ but for randomized decision trees, we show $\RMDT{\calC}=\MDT{\calC}=\MDL{\calC}=\calC$ if $\textit{mon-}\TC^0 \subseteq \textit{mon-}\calC$, and $\MDL{\AC^0}=\MDT{\AC^0}\subseteq \RMDT{\AC^0} \subseteq \AC^0$. 


\section{Preliminaries}
\label{sec:prelims}
In this section, we define basic terms and notations along with the main monotone decision tree complexity measures. The definitions of non-deterministic, randomized MDTs, and other modifications are deferred to the corresponding sections. Unless mentioned otherwise, Boolean functions discussed in this paper are from $\{0,1\}^n$ to $\{0,1\}$.
For standard definitions of Boolean circuits and related complexity classes, we refer the reader to \cite{Juk12}. \\[-3mm]

\noindent{\bf Monotonicity and alternation:} For $x\ne y \in \{0,1\}^n$, we say $x \prec y$ if $\forall i \in [n]$, $x_i \le y_i$. A \emph{chain} $\mathcal{X} $ on $\{0,1\}^n$ is a sequence $\langle x^{(1)},x^{(2)},\ldots, x^{(\ell-1)},x^{(\ell)}\rangle $ such that $\forall i\in [\ell], x^{(i)}\in\{0,1\}^n$ and $x^{(1)} \prec x^{(2)}\prec \ldots\prec x^{(\ell)}$. The alternation of a function $f$ for a chain $\mathcal{X}$, denoted as $alt(f,\mathcal{X})$ is the number of bit flips (or alternations) in the sequence $\langle f(x^{(1)}), f(x^{(2)}), \dots f(x^{(\ell)}) \rangle$. We define the \emph{alternation} of $f$ as, $\alt(f):=\max_{\text{ chain }\mathcal{X}} alt(f,\mathcal{X})$. \\[-3mm]

We say that a Boolean function is {\em monotone} if for all $x, y \in \{0,1\}^n$, $x \prec y \Rightarrow f(x) \le f(y)$. We say that a Boolean circuit is monotone if all the gates in it compute monotone functions over the respective inputs.
For any circuit complexity class $\calC$, we define $\textit{mon-}\calC \subseteq \calC$ as the class of functions which can be computed by using monotone circuits in $\calC$.\\[-3mm]

\noindent{\bf Threshold functions:} For $x \in \{0,1\}^n$, we denote the number of 1's in $x$ by ${\sf wt}(x)$. We define the ($k$-th) threshold function as, $\Th_k(x) = 1$ if ${\sf wt}(x) \ge k$ and $\Th_k(x) = 0$ otherwise.\\[-3mm]

\noindent{\bf Monotone decision trees and lists:}  We now present our generalizations of the decision tree (and list) model. Further into the paper, we introduce and study more variants by allowing non-adaptivity, randomness, restricted queries etc.

\begin{definition}[{\bf Monotone Decision Tree}]
	
	A \textit{monotone decision tree} $T$ is a rooted directed binary tree. Each internal node $v$ is labeled by a monotone function $f_v: \{0,1\}^n \to \{0,1\}$, and each leaf is labeled by a $0$ or $1$, e. Each internal node has two outgoing edges, one labeled by $0$ and another by $1$. A computation of $T$ on input $x \in \{0,1\}^n$ is the path from the root to one of the leaves $L$ that in each of the internal vertices $v$ follows the edge that has label equal to the value of
	$f_v(x)$. The label of the leaf that is reached by the path is the output of the computation.
	A monotone decision tree $T$ computes a function $f:\{0,1\}^n
	\to \{0,1\}$ if and only if on each input $x \in \{0,1\}^n$ the output of $T$ is equal to $f(x)$.
	
	The monotone decision tree complexity of $f$ is the minimum height\footnote{\emph{Height} always refers to the max. no. of internal nodes in path from root to any leaf.} of such a tree computing $f$. We denote this value by $\DTm(f)$.
\end{definition}
 
\begin{definition}[{\bf Monotone Decision List}]
The \textit{monotone decision list} model, denoted by $L= (f_1,c_1)(f_2,c_2)\dots(f_k,c_k)$ is a series of tuples $(f_i,c_i)$ where each $f_i$ is a monotone function on $n$ variables, and each $c_i$ is a Boolean constant 0 or 1. Here, each $(f_i,c_i)$ is called as a node; $f_i$ the query function of that node and $c_i$ the value of the node. The last query $f_k$ may be often assumed to be the constant function \textbf{1} w.l.o.g.
An input $x\in \{0,1\}^n$ is said to \emph{activate}\footnote{Any input activates exactly one node by this definition.} the node $(f_i,c_i)$ if $f_i(x)=1$ and $\forall j<i, f_j(x)=0$. Here $L$ is said to represent/compute the following Boolean function $f_L$ defined as:
$ f_L(x)=c_i, ~\text{where $i \in [k]$ is the unique index such that $x$ activates the $i^{th}$ node of $L$}$.

The monotone decision list complexity of a Boolean function $f$, denote by $\DLm(f)$, is the minimum size (i.e, number of nodes) of a monotone decision list computing it. 
\end{definition}

A version of decision list that has been considered in the literature is when we allow the query functions to be a simple AND of variables; called \textit{monotone term decision lists}\cite{GLR01}. When the query functions are allowed to be general (not necessarily monotone) terms, then they are called \textit{term decision lists}\cite{Bsh96} and the class of functions computed by TDLs of size at most $\poly(n)$ is denoted by {\sf TDL}.

\subsection{A Normal Form for MDLs}
\label{subsec:normal-forms}

We show, by standard arguments, that monotone decision lists can be assumed to have certain properties when the queries are allowed from any reasonably rich class of Boolean functions -- in particular, the set of  functions from $\textit{mon-}\AC^0$, $\textit{mon-}\TC^0$, $\textit{mon-}\NC^1$ etc.\\

\noindent \textsf{Property 1 - Alternating constants:} We can convert any decision list $L=(f_1,c_1)(f_2,c_2)\dots (f_k,c_k)$ computing a Boolean function $f$ on $n$ variables to an $\widetilde{L}=(\widetilde{f_1},\widetilde{c_1})(\widetilde{f_2},\widetilde{c_2})\dots (\widetilde{f_k},\widetilde{c_k})$ computing the same function where the constants $\widetilde{c_i}$'s are alternating between $0$ and $1$. 

The idea is to simply club the contiguous nodes with same $c_i$ into a single node using an OR operation over the corresponding queries (observed in e.g.~Theorem~3.1 in ~\cite{Ant02}).
\\[-4mm]

Suppose a maximal series of nodes $(f_i,c_i), (f_{i+1},c_{i+1}), \dots (f_j,c_j)$ have the same constant value (i.e, $c_i=c_{i+1}\dots c_j$). We can substitute a node $(f_i \vee f_{i+1} \dots \vee f_j,c_i)$ in place of this entire series of nodes. On application of this simplification to all the maximal contiguous nodes with identical constants, we finally get an equivalent monotone decision list with alternating constants. 

To observe that this transformation does not affect the output, we argue that the following simplification holds. If $(f_1,c)(f_2,c)$ are two consecutive nodes in a decision list, the both of them can be replaced with the single node $(f_1 \vee f_2,c)$ to get an equivalent decision list. For this, we note that if on an input $x$, both $f_1$ and $f_2$ fail (evaluate to 0), then since neither of the original two nodes would have been activated and neither does the new node $(f_1\vee f_2,c)$ the replacement does not alter the output returned by the decision list. On the other hand, say the node $(f_1,c)$ is the activated node. This means all the queries before $f_1$ would have evaluated to 0. Since these nodes remain unaltered, and $f_1\vee f_2$ would pass on $x$, the node $(f_1\vee f_2,c)$ is the activated one, and therefore returns $c$, which is the same value returned by the original decision list. The same argument holds when $(f_2,c)$ is the node that $x$ activates in the original decision list.\\




\noindent \textsf{Property 2 - Forward firing:} By forward firing, we mean that on any input $x \in \zon$, if certain query of a decision list passes (i.e., evaluates to 1), then so do all the queries that follow it. 

We claim that the decision list $\widetilde{L}=(g_1,c_1)(g_2,c_2)\dots (g_k,c_k)$, where $g_i= \bigvee_{j=1}^{i} f_j$, represents $f$ and has the forward firing property. First, note that since $g_{i+1}\equiv g_i \vee f_{i+1}$, we immediately have that $g_{i} \Rightarrow g_{i+1}$ holds true. Now to show that $\widetilde{L}$ is equivalent to $L$, suppose that on an input $x$, the $t^{th}$ node is activated in $L$. Then note that $g_t(x)=\bigvee_{j=1}^{t} f_j(x)=1$, since $f_t(x)=1$ by the definition of $t$. For $i<t$, we have $g_i(x)=\bigvee_{j=1}^{i} f_j(x)=0$, since each $f_j$ is failed for $j<t$ as $x$ failed all the first $t-1$ queries of $L$. Therefore, $x$ indeed activates the $t^{th}$ node of $\widetilde{L}$ too, hence $\widetilde{L}$ on input $x$ outputs $c_t=f(x)$.

It is easy to note that the procedure for ensuring Property 2 does not disturb Property 1 if the decision list already has it.  That is, there exists a decision list with both the properties.
The above normal forms are invoked in a context that if the $f_i$'s are from a circuit complexity class that is closed under taking OR of polynomially many bits, and when $k$ is polynomial in $n$, then the new  query functions also belong to that class (they admit monotone circuits in that class).
\\[-3mm]



Recalling the definition of \emph{alternation} of a function from Section~\ref{sec:prelims}, we state the following characterization of Boolean functions originally proved in \cite{BCOST15}.
\begin{lemma}[\textbf{Characterization of Alternation~\cite{BCOST15}}] 
\label{lem:alt-decomp}
{ For any $f:\zon \to \zo$ there exists $k = \alt(f)$ monotone functions $f_1, \dots, f_k$ each from $\{0,1\}^n$ to $\{0,1\}$ such that}
	
	\[f(x) = \begin{cases}
		 \oplus_{i=1}^k f_i(x) & \text{, if } f(0^n) = 0 \\
		 \neg \oplus_{i=1}^k f_i(x) & \text{, if } f(0^n) = 1.
		\end{cases}
	\]
\end{lemma}

\section{Our Tool: Monotone Decomposition of Boolean Functions}
\label{sec:mon-decom}
Motivated by the characterization stated in previous section we define a  monotone decomposition of a Boolean function as follows. This notion will be helpful while analyzing the monotone decision tree complexity measures.

\begin{definition}[{\bf Monotone Decomposition of a Boolean function}]
For any Boolean function $f:\zon \to \zo$, a {\em monotone decomposition} of $f$ is a minimal set of monotone functions $M = \{ f_1, f_2, \ldots f_k \}$ such that $f = \oplus_{i \in [k]} f_i$ --
here $k$ is said to be the length of the decomposition. We call each $f_i$ to be a {\em monotone component} in the decomposition. 

We consider two variants of monotone decompositions obtained by imposing additional constraints:
\begin{description}
\item{\sf Implication property:} There exists an ordering of $M$ such that $\forall i \in [k-1],$ $f_{i} \Rightarrow f_{i+1}$ holds. In this case, the functions are called {\em boundary functions} of the decomposition of $f$. We call monotone decompositions that satisfy this property as {\em boundary decompositions} of $f$.
\item{\sf Optimality Property:} If the set $M$ is also of minimum size monotone decomposition of $f$. That is, there does not exist fewer set of monotone functions whose parity is the given function $f$. We call monotone decompositions that satisfy this property as {\em alternation decompositions} of the function $f$.
\end{description}
\end{definition}

The proof of decomposition of Boolean functions into parity of monotone functions in \cite{BCOST15} actually implies a monotone decomposition of length $\alt(f)$ (or $\alt(f)+1$ if $f(0^n)=1$) which has the optimality and implication properties. Their proof also implies that if optimality property holds for a monotone decomposition, then the length of the decomposition must necessarily be equal to $\alt(f)$ or $\alt(f)+1$. Because of this reason, we call decompositions with optimality property as an {\em alternation decomposition}. We now state the following lemma (already proved in \cite{BCOST15}) bringing out the details to substantiate the extra properties that we need -- we present its proof in Appendix~\ref{app:lem:alternation-decomp}.

\begin{lemma}[{\bf Monotone Decomposition Lemma}]
\label{lem:alternation-decomp}
For any Boolean function $f:\zon \to \zo$ there is a monotone decomposition with implication and optimality properties of length $\alt(f)$ (if $f(0^n) = 0$) and length $\alt(f)+1$ (if $f(0^n)=1$).
\end{lemma}

We will often use the fact that monotone decomposition with implication property corresponds to a monotone decision list:

\begin{proposition}
	\label{prop:decomp-dl}
	Let $k$ be an even number\footnote{This is a minor constraint as we can prepend or append constant functions to a monotone decomposition or an MDL.} and $f_1 \Rightarrow \dots \Rightarrow f_k$ be Boolean functions. The following functions are equivalent to one another:
	
	\begin{itemize}
		\item $f_1 \oplus f_2 \oplus \dots \oplus f_k$,
		\item $(f_1,0)(f_2,1)\dots(f_k,1)(\textbf{1},0)$,
		\item $\overline{f_1}f_2 \vee \overline{f_3}f_4 \vee \dots \vee \overline{f_{k-1}}f_k$,
		\item $(\textbf{0} \vee \overline{f_1})\wedge(f_2 \vee \overline{f_3})\dots \wedge (f_{k-2} \vee \overline{f_{k-1}})\wedge (f_k \vee \overline{\textbf{1}})$.
	\end{itemize}
\end{proposition}

\begin{proof}
	Because of the implication property of $f_i$'s, for any input $x$, the sequence of bits $s:=\langle f_1(x),f_2(x),\dots,f_k(x) \rangle$ is sorted (from 0 to 1). All the above four functions compute whether the number of 1's in the sequence $s$ is odd. To see this, first check that all the functions evaluate to 0 if $s=1^k$. As we keep flipping the last 1-bit of $s$ to 0, the function value (of each of the above four functions) alternates between 0 and 1.
\end{proof}

\begin{aproof}{Proof of Monotone Decomposition Lemma (Lemma~\ref{lem:alternation-decomp})}{app:lem:alternation-decomp}
We reproduce the proof of Lemma 1 in \cite{BCOST15} bringing out the details to substantiate the extra properties that we need.
Recall that number of alternations of a chain $\mathcal{X}$ w.r.t. a function $f$, $\alt(f,X)$ is the number of times the value of $f$ changes as we move along $\mathcal{X}$. 
	
	For an input $x\in \{0,1\}^n$, we define as $a_f(x):=max\{ \alt(f,\mathcal{X})\}$, where  $\mathcal{X}$ is any chain starting at $x$. Note that $\alt(f)=max_x \{a_f(x)\}=a_f(0^n)$. We use induction on $k=\alt(f)$ to prove the lemma. 
	
	The base case $\alt(f)=0$ means $f$ is either 0 or 1 for all inputs. When $f(0^n)=0$, we have $f(x)=0$ and when $f(0^n)=1$, we have $f(x)=1=\neg (0)$. We thus can trivially decompose $f$ into $k=0$ many monotone functions.
	
	Now we assume that the claim is true for all functions on $n$ variables with alternation less than $k$ and prove it for the function $f$ with $\alt(f)=k$. For $1 \le i \le k$, we define the following Boolean function on $n$ variables: $$g_i(x)=1 \iff a_f(x) < i.$$ We shall show that these $g_i$'s are actually the desired monotone components $f_i$'s. Firstly note that all $g_i$'s are monotone as for any pair of inputs $x \prec y$, we have $a_f(y)\le a_f(x)$ and hence $g_i(x)\le g_i(y)$. Also, the implication property holds because $g_i(x)=1 \Rightarrow a_f(x)<i<i+1 \Rightarrow g_{i+1}(x)=1$.
	
	We will argue that $f$ is indeed equal to $g_1 \oplus g_2 \oplus \dots \oplus g_k$ when $f(0^n)=0$. A similar derivation can be done for $f(0^n)=1$ case. Consider the function $f':=f\vee \overline{g_k}$. We will prove the following properties of $f'$.
	
	\begin{itemize}
	\item{$f'(0^n)=1$: As $\alt(f)=k$, there would be at least one chain (call $\mathcal{X}^*$) whose alternation w.r.t.~$f$ is $k$. Hence, $a_f(0^n)=a(f,\mathcal{X}^*)=k$ and $g_k(0^n)=0$, which implies $f'(0^n)=f(0^n) \vee \overline{g_k(0^n)}=1$. }
	
	\item{$\alt(f')=k-1$}: We will argue that for any chain $\mathcal{X}^*\equiv \langle x^{(1)},x^{(2)},\dots x^{(l)} \rangle$, if $\alt(f,\mathcal{X}^*)=k$, then $\alt(f',\mathcal{X}^*)=k-1$. Note that as $f(0^n)=0$, there exists a first index $p$ such that $f(x^{(p)})=1$ (as $k>0$). By definition of $g_k$, each of $x^{(1)},\dots ,x^{(p-1)}$ does not satisfy $g_k$ (for any $1\le i \le p-1$, the suffix-chain $\calC_i$ of $\mathcal{X}^*$ starting at $x^{(i)}$ has $\alt(f,\calC_i)=k$, thereby making $g_k(x^{(i)})=0$). 
	
	Thus the values of $f'=f\vee \overline{g_k}$ become 1 for the first $p-1$ inputs. Since we know $f(x^{(p)})=1$ and $\alt(f,\mathcal{X}^*)=k$, we get that $\alt(f',\mathcal{X}^*)=k-1$. We can also argue that there is no chain with alternations more than $k-1$ w.r.t.~$f'$. Therefore, $\alt(f')=k-1$.
	
	As $\alt(f')<k$, by induction hypothesis we obtain the functions $g'_1,g'_2,\dots ,g'_{k-1}$ with implication property and $f'=\neg (g'_1\oplus g'_2 \oplus\dots \oplus g'_{k-1})$.
	
	\item{$\forall i\in [k-1]$, $g'_i \equiv g_i$: If an input $x$ satisfies $g_k(x)=0$, it means $a_f(x)=k$, and hence $g_i(x)=0$ for $i\in [k-1]$. In the above part notice that we showed $\alt(f',\mathcal{X})=k-1$ when $\alt(f,\mathcal{X})=k$. This means $g'_i(x)=0$ too for all $i\in [k-1]$. 
		
	On the other hand, for inputs $x$ such that $g_k(x)=1$, observe that $f'(x)=f(x)$. As $a_f(x)$ and $a_{f'}(x)$ depend on the values of the corresponding functions only ``above'' $x$ and the functions $g_i$ and $g'_i$ are defined based on these values, they must be equal. }
	
	Hence for all $x$, we have $g'_i(x)=g_i(x)$. By our definition of $g_k$, it can be shown that $f\Rightarrow g_k$ is a tautology. This fact along with $f'=f\vee \overline{g_k}$ implies that $f\equiv \overline{f'} \oplus g_k=g'_1 \oplus \dots g'_{k-1} \oplus g_k=g_1 \oplus \dots \oplus g_k$.
	\end{itemize}
\end{aproof}

\subsection{Uniqueness of Alternation Decomposition in a Special Case}

For a given function, there could be multiple monotone decompositions with implication and optimality properties. 
For example, consider the function $f = (x = 0^{n/2}1^{n/2})$ for any even $n > 2$. The function value is $1$ if and only if the input is $0^{n/2}1^{n/2}$. Clearly, the alternation of this function is $2$. We will give two different decompositions for this function. One is $f(x) = \Th_{n/2}(x)\oplus (\Th_{n/2}(x) \land (x \neq 0^{n/2}1^{n/2}))$. Another decomposition is $f(x) = \Th_{n/2 + 1}(x) \oplus (\Th_{n/2 + 1}(x) \vee (x = 0^{n/2}1^{n/2}))$. It is easy to verify that these two monotone decompositions satisfy the additional properties. However, we do have a unique monotone decomposition with implication and optimality properties in a special case:

\begin{proposition}
\label{prop:unique-alternation-decomp}
If a Boolean function $f$ exhibits uniform alternation for all chains of maximal length in the Boolean hypercube, then $f$ has a unique alternation decomposition with implication property.
\end{proposition}
\begin{proof}
Suppose that there are two alternation decompositions for $f$, both having implication property. Denote them by $\{f_1,f_2, \ldots f_k\}$ and $\{f_1',f_2', \ldots f_k'\}$ such that $f_{1} \Rightarrow f_{2} \Rightarrow \ldots f_{k-1} \Rightarrow f_k$ and $f_{1}' \Rightarrow f_{2}' \Rightarrow \ldots f_{k-1}' \Rightarrow f_k'$ and $k=\alt(f)$ (assuming $f(0^n)=0$). We will argue by contradiction that the set of functions must be the same i.e., $f_i = f'_i$ for all $i$. Consider the largest $i$ for which there exists $x \in \zon$ such that $f_i(x) \ne f_i'(x)$. Without loss of generality, suppose $f_i(x) = 1$ and $f'_i(x)=0$. Fix such an $x$ and a chain $0^n = x^{(0)} \prec x^{(1)} \prec x^{(2)} \prec \dots \prec x^{(n)}=1^n$ containing $x$ such that all the inputs $y$ before $x$ in the above chain satisfy $f_i(y)=0$ (so that $f_{i-1}(y)=0$ as well).
By using the optimality property, we have $f = f_1 \oplus f_2 \dots \oplus f_k = f'_1 \oplus f'_2 \dots \oplus f'_k$. 

If $i = 1$, then the two decompositions do  not agree on the input $x$ as $f_i = f'_i$ for $i>1$, which is a contradiction. 

If $i>1$, we must necessarily have that $f_{i-1}(x) = 1$, as otherwise $1=f_1 \oplus \dots \oplus f_{i} = f'_1 \oplus \dots \oplus f'_{i}=0$ (as  $i$ is the largest index such that $f_i \ne f'_i$ and due to implication property). Therefore, $f_i$ and $f_{i-1}$ both change their evaluation from 0 to 1 on changing the input to $x$ from the input that is just before $x$ in the chain we considered. This contradicts the fact that this chain has $k$ alternations.
\end{proof}

\subsection{Constraints on Monotone Components for $\TC^0$ and beyond}
We continue the study in this section by imposing complexity constraints on the function $f$. A natural question to ask is if the monotone components of $f$ in its monotone decomposition  are necessarily harder than $f$ in terms of circuit complexity classes. 
We first answer this question for classes that contain monotone circuits for the threshold functions. In this case, we show that we can always find a monotone decomposition where the component functions are in $\textit{mon-}\calC$.

\begin{lemma}
\label{lem:tc0-decomp}
If $\textit{mon-}\TC^0 \subseteq \textit{mon-}\calC$, then for any $f$ computed by a circuit in the class $\calC$, there is a monotone decomposition $f_1 \oplus  f_2 \oplus \ldots \oplus f_{2n+1}$ with implication property such that each $f_i$ is in $\textit{mon-}\calC$.
\end{lemma}
\begin{proof}
Recall that $\Th_k$ denotes the function $\Th_k(x)=1$ iff the number of $1$'s in $x$ (denoted by $\mathsf{wt}(x)$) is at least $k$.
We directly write the decomposition.
For each $1 \le i \le 2n+1$,
$f_i(x) = \Th_{k+1}(x) \vee (\Th_{k}(x) \wedge f(x))$ when $i$ is $2k+1$ else $f_i(x) = \Th_{k}(x)$ when $i$ is $2k$.
%

\noindent{\bf Correctness:} Let $w = \mathsf{wt}(x)$. For any $x \in \zon$:

\begin{description}
\item{\sf For $i < w$:} we have $\Th_i(x)=1$ and $\Th_{i+1}(x) \vee (\Th_{i}(x) \wedge f(x)) = 1$.
\item{\sf For $i > w$:} we have $\Th_i(x)=0$ and $\Th_{i+1}(x) \vee (\Th_{i}(x) \wedge f(x)) = 0$.
\item{\sf For $i = w$:} we have $\Th_i(x)=1$ and $\Th_{i+1}(x) \vee (\Th_{i}(x) \wedge f(x)) = f(x)$
\end{description}
\noindent Using this we can compute the boundary functions as follows:
\begin{eqnarray*}
\forall i, 1\le i \le 2w:~f_i(x) & = & 1 \textrm{ and } \forall i, 2w+2\le i \le 2n+1:~ f_i(x) = 0\\[-1mm]
f_{2w+1}(x) & = & \Th_{w+1}(x) \vee (\Th_{w}(x) \wedge f(x)) = f(x)
\end{eqnarray*}

\noindent Observing the above evaluations, note that the implication property holds for $f_i$'s as $f_{2n+1} \Rightarrow f_{2n} \Rightarrow \dots \Rightarrow f_2 \Rightarrow f_1$. The expression $f_1(x) \oplus f_2(x) \dots \oplus f_{2n+1}(x)$ evaluates to $1 \oplus 1 \dots 1 \oplus f(x) \oplus 0 \dots \oplus 0$ where number of $1$'s before $f(x)$ in the above expression is $2w$ (which is even). Thus, on a given input $x$, the decomposition evaluates to $f(x)$.\\[-3mm]

\noindent{\bf Complexity bound for $f_i$:}	Now we will prove that the query functions actually are in $\textit{mon-}\calC$. The $f_i$'s for even $i$, being threshold functions, satisfy this property. We now show that $f_{2k+1}=\Th_{k+1} \vee (\Th_k \wedge f)$ has a monotone circuit in $\calC$ circuit for $0\le k \le n$. 

Let $C$ be a circuit in $\calC$ computing $f$. We first push down all the negation gates to the input variables in $C$, by using the fact that $\neg({\Th_\ell(a_1,a_2,\dots a_m)})=\Th_{m-\ell+1}(\neg a_1,\neg a_2,\dots \neg a_m)$. This can be done with only polynomial increase in size of $C$, and same depth. In the resulting circuit, further remove the negations at the variables as follows: if $\bar{x_j}$ appears, replace it with $\Th_k(\{x_1,x_2,\dots x_n\}\setminus \{x_j\})$. Call the final circuit $C'$. Note that $C'$ does not have any negation gates, and therefore is a monotone circuit in $\calC$. We shall argue that the circuit $C'' = \Th_{k+1} \vee (\Th_k \wedge C')$ computes $f_{2k+1}$. It can be seen that $C''$ outputs $0$ for inputs of weight less than $k$ and outputs 1 for inputs of weight more than $k$, which is exactly what $f_{2k+1}$ evaluates to on these inputs. Therefore, it suffices to show that $C''$ correctly outputs $f_{2k+1}(x)=f(x)$ on inputs of weight exactly equal to $k$. To see this, note that the final transformation of $\bar{x_j}=\Th_k(\{x_1,x_2,,\dots x_n\}\setminus \{x_j\})$ holds when $\mathsf{wt}(x)=k$.
\end{proof}

By the above lemma, for any $f\in \TC^0$ we have given a decomposition into $2n+1$ monotone $\TC^0$ functions. However, the monotone decomposition can be far from having the optimality property. If the function has uniform alternation among all chains of length $n+1$, then this can be improved to $\alt(f)$ keeping the complexity of the monotone components to be within $\TC^0$, as show below.\\

\noindent {\bf Monotone decomposition for `special' functions in $\TC^0$:}
\begin{proposition}
\label{prop:equal-alt-decomp}
If $f\in \TC^0$ has uniform alternation across all chains of length $n+1$, say $\mathsf{alt}(f)=k$, then there is a monotone decomposition of $f$ of length $k$ or $k+1$ where all the components are in $\TC^0$.
\end{proposition}
\begin{proof}	
		It is worthwhile to recall that the decomposition is unique when all the chains of length $n+1$ have same alternation (Proposition~\ref{prop:unique-alternation-decomp}). We will give the proof only when $f(0^n)=0$ (decomposition into $k$ monotone $\TC^0$ components). In the other case, we may just take the negation of the decomposition of the complement function (which would still be in $\TC^0$ and have uniform alternation). 
		
		Let $f\equiv f_1 \oplus f_2 \oplus \dots \oplus f_k$ where we have $k=\alt(f)$ and $\forall i\in[k-1]:~f_{i} \Rightarrow f_{i+1}$. We exploit the following structure of $f_i$ to obtain a $\TC^0$ circuit computing it. This observation comes from the uniform alternation feature of $f$.
		
		\begin{observation}
			For any $1\le i\le k$, $f_i(x)=1$ iff there is a chain from $0^n$ to $x$ with at least $k-i+1$ alternations. 
		\end{observation}
		
		Since all chains of length $\mathsf{wt}(x)+1$ from $0^n$ to $x$ have same alternations, we can take any arbitrary such chain to decide the value of $f_i(x)$. We will be considering the chain of inputs obtained by repeatedly making the leftmost 1 to 0.
		
		For input $x\in \{0,1\}^n$, we define $n$ new strings $y^{(i)}$ for $1\le i \le n$ inductively. The string $y^{(1)}$ is obtained by making the leftmost 1 of $x$ into a 0; and similarly for any other $i>1$, we obtain $y^{(i)}$ by making the leftmost 1 of $y^{(i-1)}$ into a 0. If the number of 1's in $x$ happens to be lesser than $i$ then we define $y^{(i)}$ to be $0^n$. We will now argue that $y^{(i)}$'s can be computed in $\TC^0$. Since {\sc Bitcount} $\in \TC^0$, we can obtain all the prefix sums $p^{(i)}\in \{0,1\}^{\log n}=\sum_{j=1}^{j=i} x_j$ in $\TC^0$. Then, the $j^{th}$ bit of $y^{(i)}$, namely $y^{(i)}_j:=x^{(i)}_j \wedge (p^{(j)} > i)$. This works because for the first $i$ 1's of $x$, the prefix sum $p^{(i)}$ is at most $i$, and it is greater than $i$ for all other bits. Also, this can be implemented in $\TC^0$ as we know that the relation $>$ is in $\AC^0$. 
		
		Once we have obtained $y^{(i)}$'s, we construct the Boolean sequence $s$ of length $n+1$, $s:=f(x).f(y_1).f(y_2).\dots f(y_n)$. This also can be done in $\TC^0$ as we know $f\in \TC^0$ and each of $y^{(i)}$'s can be parallelly evaluated in $\TC^0$. Notice that the sequence $\langle y^{(n)},y^{(n-1)},\dots y^{(1)},x \rangle$ represents the same input $0^n$ till a point and represents a chain from $0^n$ to $x$ from the next point. Hence by our observation and the fact that $f(0^n)=0$, we note that $f_i(x)=1$ iff $alternations(s)\ge i$. Now to count the number of alternations in the sequence $s$, we compare every two successive bits of $s$ and check whether the number of times it changes is at least $i$. That is, $f_i \equiv \Th_{k-i+1}(s_1\oplus s_2, s_2\oplus s_3, \dots ,s_n\oplus s_{n+1})$. Clearly this can also be done in $\TC^0$ as bit-parity and threshold gates are in $\TC^0$.
		
		It has to be noted that we only gave a decomposition with monotone functions in $\TC^0$. We do not know whether these functions have monotone $\TC^0$ circuits. The other limitation of course, is that it only works for functions with uniform alternation.
\end{proof}

\section{Deterministic Adaptive Monotone Decision Trees}
\label{sec:d-mdt}

The model that we study first is the deterministic adaptive monotone decision trees and the associated complexity measures like $\DTm $ and $\DLm$ defined in Introduction.

\subsection{Monotone Decision Trees and Monotone Decision Lists}
In this section, we relate complexity of monotone decision trees and monotone decision lists. In particular, we prove the following theorem:

\begin{theorem}
\label{thm:mdt-mdl}
For any Boolean function $f$, $\DTm(f) = \lceil\log\DLm(f)\rceil$ \footnote{Using the same proof, we also get for any circuit complexity class $\calC$ with appropriate closure properties that $\MDT{\calC} = \MDL{\calC}$ -- this will be used in Section~\ref{sec:mdt-query-restrict}.}. 
\end{theorem}
\begin{proof}
($\ge$): To show $\DLm(f)\le 2^{\DTm(f)}$, note that it suffices to show that, from a monotone decision tree we can construct a monotone decision list of the same size. As the number of leaves in the optimal tree is upper bounded by $2^{\DTm(f)}$, the inequality then follows. Now, going to proving this relation itself, it follows from a known construction given by Blum~\cite{Blu92}. Since we need it to work in the context of monotone queries as well, we provide a self-contained argument in the following lemma.
\begin{lemma}
\label{lem:mdt-to-mdl}
	Let $T$ be a monotone decision tree of size $k$ computing a function $f$ on $n$ variables. Then there is a monotone decision list of size $k$ computing $f$.
\end{lemma}
\begin{proof}

	 For each leaf $\ell$ in $T$, denote the path from source/root to $\ell$ as a string $s_\ell \in \{0,1\}^*$ defined as $s_\ell[i]=1$ iff the $i^{th}$ edge in the path is positive labeled.
	
	Let $S=s_1,s_2,\dots ,s_k$ be the ordering of all such strings  lexicographically, from larger to smaller. This is equivalent to ordering the leaves of the tree from right to left (under the convention that the left-edges correspond to 0 and the right-edges to 1 at each query in the decision tree). 
	
	We claim that the following monotone decision list $L$ computes $f$:
	$$L = (t_1,label({\ell_1}))(t_2,label({\ell_2}))\dots (t_k,label({\ell_k})),$$
	
	where $label(\ell_i)\in \zo$ denotes the label of the leaf $\ell_i$ and $t_i$ is the function obtained by the conjunction of all \textit{passed queries} along the path from root of $T$ to $\ell_i$ (see Figure~\ref{fig:2}). 
	
	To show this, let an input $x$ activate the $i^{th}$ node of $L$. We have $t_i(x)=1$. Now it suffices to prove that the computation by $T$ on $x$ reaches the leaf $\ell_i$ and hence outputs the same value $label(\ell_i)$.
	
	To prove this by contradiction, suppose not; that is, let the leaf reached in the decision tree on input $x$ be $\ell_j \ne \ell_i$.  Let the first position in which corresponding paths (strings) $s_{j}$ and $s_{i}$ differ be at the $k$-th index. Such a position exists as otherwise one string is a prefix of the other, which is impossible as these binary strings `encode' the paths from root to leaves in the tree.
	
	\begin{description}
		\item{\textsf{Case(1) - $s_{j}[k]=0$ and $s_{i}[k]=1$}}: Since the strings are same till the $(k-1)^{th}$ position, so are the corresponding paths in the decision tree. Let the least-common-ancestor node of $\ell_i$ and $\ell_j$ be labeled by a (monotone) query $f_1$ (w.l.o.g). Notice that since $s_{i}[k]=1$ it makes $f_1$ a passed query along the path to $\ell_i$ and so $f_1\in t_{i}$ (i.e., $t_i$ is an AND of $f_1$ and other functions, by the definition of $t_i$). Since we have that $t_i(x)=1$, it implies $f_1(x)=1$ which means $s_{j}[k]=1$, a contradiction.
		\item{\textsf{Case(2) - $s_{j}[k]=1$ and $s_{i}[k]=0$}}: This means $s_{j}>s_{i}$. Note that since the computation in the decision tree reaches $\ell_j$ on $x$ and $t_{j}$ is a conjunction of passed queries in the path to $\ell_j$, we must have $t_{j}(x)=1$. And since $s_{j}>s_{i}$, it will happen that the query $t_{j}$ appears before $t_{i}$ in $L$. Hence, the node $(t_i, label(\ell_i))$ would not be activated on $x$, contradicting the definition of $i$ itself.
	\end{description}
	Hence, the proof of the lemma follows.

\end{proof}

\noindent ($\mathsf{\le}$): To now show that $\DTm(f)$ is at most $\lceil\log\DLm(f)\rceil$, we construct a monotone decision tree of height $\lceil\log k\rceil$ given an equivalent monotone decision list of length $k$. This is proved in the following lemma.
\begin{lemma}
	Let $L=(f_1,c_1)(f_2,c_2)\dots (f_{k-1},c_{k-1})(\textbf{1},c_k)$ be a monotone decision list for the Boolean function $f$ on $n$ variables. Then there is a monotone decision tree $T$ of height $\lceil\log k\rceil$ computing $f$.
\end{lemma}
\begin{proof}
Without loss of generality, we assume that the decision list $L$ is in the normal form (Sub-section~\ref{subsec:normal-forms}). On an input $x \in \zon$, there exists exactly one node where the corresponding query and all the queries to its right pass, and all the queries to its left fail. A natural approach to search for this switch point (the activated node) algorithmically by querying the corresponding functions, is to do a binary search and return the $c_i$ corresponding to the index resulting from the search. We shall use this idea to construct a decision tree $T$ for $f$. 

At the root of $T$, we query $f_{\lceil(1+k)/2\rceil}$. If it is zero (go left in the tree), it means the activated node is to the right of it, otherwise (go right in the tree) to its left (or itself). 
This way we can give a recursive construction for $T$ (Use the right half of $L$ on the left branch and the left half of $L$ on the right branch). To obtain the labels for the leaves of $T$, we write down $c_1,c_2\dots c_k$ from the right to left.

\begin{figure}
\begin{center}	
\begin{forest}
	BDT
	[$f_4$
	[$f_6$
	[$f_7$
	[$c_8$]
	[$c_7$]
	]
	[$f_5$
	[$c_6$]
	[$c_5$]
	]
	]
	[$f_2$
	[$f_3$
	[$c_4$]
	[$c_3$]
	]
	[$f_1$
	[$c_2$]
	[$c_1$]
	]
	]
	]
\end{forest}
\end{center}
\caption{MDT corresponding to the MDL $(f_1,c_1)(f_2,c_2)\dots(f_7,c_7)(1,c_8)$}
\label{fig:1}
\end{figure}
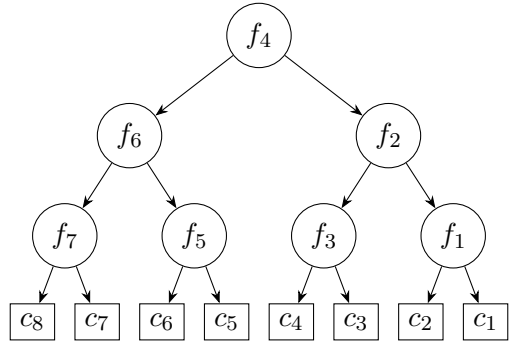

\begin{figure}
	\begin{center}	
		\begin{forest}
			BDT
			[$f_1$
			[$f_2$
			[$f_4$
			[$c_1$]
			[$c_2$]
			]
			[$f_5$
			[$c_3$]
			[$c_4$]
			]
			]
			[$f_3$
			[$f_6$
			[$c_5$]
			[$c_6$]
			]
			[$f_7$
			[$c_7$]
			[$c_8$]
			]
			]
			]
		\end{forest}
	\end{center}
	\caption{MDL corresponding to the above MDT is $(f_1.f_3.f_7, c_8)(f_1.f_3,c_7)(f_1.f_6,c_6)(f_1,c_5)\dots(1,c_1)$}
	\label{fig:2}
\end{figure}
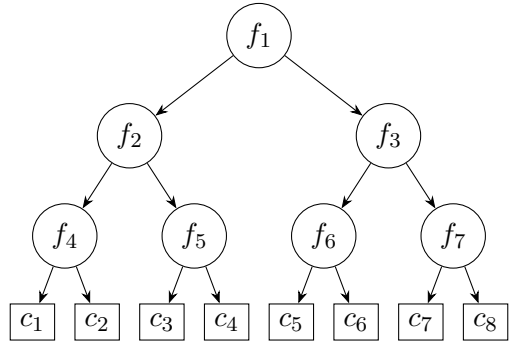

\vspace{2mm}	

As an example, if $L=(f_1,c_1)(f_2,c_2)(f_3,c_3)(f_4,c_4)(f_5,c_5)(f_6,c_6)(f_7,c_7)(1, c_8)$ , the constructed tree $T$ would be as shown in Figure~\ref{fig:1}. It can be seen that the height of $T$ would be $\lceil\log k\rceil$, since at each query the nodes of the residual decision list are nearly distributed equally onto left and right sides. Finally, we note that the queries in $T$ are identical to the queries in $L$, and hence are monotone. 
\end{proof}
This completes the proof of the relation between monotone decision list complexity and monotone decision tree complexity.
\end{proof}

\subsection{Characterizing Adaptive Decision Tree Height} 
We will now prove the main theorem of this section which characterizes $\DTm(f)$ (and en route $\DLm(f)$ too) in terms of $\alt(f)$.
\begin{theorem}
\label{thm:mdt-alt}
For any Boolean function $f$,
$\DTm(f) = \lceil\log (\alt(f)+1)\rceil$.
\end{theorem}
\begin{proof}

It suffices to show that $\DLm(f) = \alt(f)+1$ because of Theorem~\ref{thm:mdt-mdl}.

\begin{description}
\item{\sf $\DLm(f) \le \alt(f)+1$:}
First, suppose $f(0^n)=0$. Then by Lemma~\ref{lem:alternation-decomp} there are $k=\alt(f)$  many monotone functions such that $f=f_1 \oplus f_2 \oplus \dots \oplus f_k$, and $\forall i, f_i \Rightarrow f_{i+1}$. It can be shown easily that the monotone decision list $(f_1,0)(f_2,1)(f_3,0)(f_4,1)\dots (f_k,1)(\textbf{1},0)$ or $(f_1,1)(f_2,0)(f_3,1)(f_4,0)\dots (f_k,1)(\textbf{1},0)$ computes $f$, when $k$ is even or odd respectively. On the other hand, if $f(0^n)=1$, we have $f=f_1 \oplus f_2 \oplus \dots \oplus f_k \oplus 1$, which gives the monotone decision lists $(f_1,1)(f_2,0)\dots (\textbf{1},1)$ or $(f_1,0)(f_2,1)\dots (\textbf{1},1)$ computing $f$ depending on whether $k$ is even or odd respectively.
\item{\sf $\DLm(f) \ge \alt(f)+1$:} We claim that 
if a Boolean function $f$ on $n$ variables can be computed by a monotone decision list $L=(f_1,c_1)(f_2,c_2)\dots (f_\ell,c_\ell)$ of length $\ell$, we have $\alt(f) \le \ell - 1$.
To show this, it suffices to argue that for any chain $ x^{(1)} \prec x^{(2)} \prec \ldots \prec x^{(s)} $ in the Boolean hypercube, where $1\le s\le n+1$; the number of alternations of the function $f$ along the chain is at most $\ell - 1$. 


Consider the sequence $S$ of length $s$ where for $1\le i\le s$, the integer $S[i]$ is the index of the node activated on inputting $x^{(i)}$ to $L$. Hence, $1\leq S[i] \le \ell$ for every $i$. By definition of the activated node, observe that for any $1\le i<s$, $f_{S[i]}(x^{(i)})=1$, which implies $f_{S[i]}(x^{(i+1)})=1$ too, since $x^{(i)} \prec x^{(i+1)}$. This implies that the node that $x^{(i+1)}$ activates cannot be after $f_{S[i]}$. That is, $S[i+1]\le S[i]$ for all $1\le i<s$. If two consecutive elements in chain activate the same node, $L$ outputs the same value on these assignments and hence there is no alternation at that point of the chain. Thus, the number of alternations of the function $f$ along this chain is upper bounded by the number pairs $(S[i],S[i+1])$ such that $S[i] \ne S[i+1]$. Since $1\leq S[i]\leq \ell$ and $S[i]\geq S[i+1]$ for all $i$, we get $\mathsf{alt}(f) \le \ell-1$.

\end{description}
\end{proof}



\subsection{Constructing Adaptive MDTs from Negation Limited Circuits}

The above theorem provides a characterization for decision tree height in terms of alternation $\alt(f)$ of the Boolean function. A classical result by Markov\cite{Markov58}, implies that any Boolean function can be computed by Boolean circuits that use at most $\lceil \log(\alt(f)+1)\rceil$ many negation gates. Since the number of negation gates in the circuit can be logarithmically smaller, the following result is interesting.


\begin{theorem}
\label{prop:anadt-neg}
Let $f$ be a Boolean function computed by a circuit $C$ using $k$ negations. Then there is a monotone decision tree of height $k+1$ computing $f$.
\end{theorem}
\begin{proof}
Call the bottom-most negation gate in $C$ as $g$. The input to $g$ would be computed by a monotone circuit (call it $C_+$). The first query in our MDT shall be the function computed by $C_+$. Let $C_0$ and $C_1$ be the circuits obtained upon replacing\footnote{using 0 and 1 (respectively) whenever the output of $g$ is required} $g$ by the constants 0 and 1 respectively. Note that these circuits have $k-1$ negation gates. On the $C_+(x)=0$ branch of the decision tree, we will use $C_1$ and on the other side we use $C_0$ instead of $C$ to obtain the queries at the second level of the decision tree. We keep applying the same principle as we did to $C$ until there are no negations left, when we can directly query the entire function to obtain the return value. In any branch of the tree, at each query, the height of the tree increases by 1 while the number of negations in the residual circuit decreases by 1. Therefore, the height of the constructed monotone decision tree is $k+1$.
\end{proof}

%


\section{Deterministic Non-Adaptive Monotone Decision Trees}
\label{sec:d-namdt}

We establish a relation between the non-adaptive monotone decision tree and alternation:

\begin{theorem}
\label{thm:nadt-alt}
For any Boolean function, $\alt(f)=k$ if and only if $f$ can be computed by a non-adaptive monotone decision tree of height $k$.
\end{theorem}

\begin{proof}
	To outline the idea used here -- for the forward implication we use Lemma~\ref{lem:alternation-decomp} to design the decision tree. For the reverse implication, since the tree is non-adaptive, the query function at each level of the decision tree will be the same. Using this fact, we argue that any chain  must have alternation at most $k$ with respect to $f$. To be more detailed,
	
($\Rightarrow$): Suppose $\alt(f)=k$. 
Applying Lemma~\ref{lem:alternation-decomp}, we describe the non-adaptive monotone decision tree of height $k$ is as follows. At level $i$ (root being level 1), all nodes will query the function value $f_i$. For labeling the leaves, simply label each leaf by the parity of the results of the queries along the path from root of the tree to that leaf. This way, it is non-adaptive by definition and in each path, the values of all $f_i(x)$'s are known and can compute the value of $f(x)$ as described above.

($\Leftarrow$):
Let $f$ be computed by decision tree $T$ of height $k$. Since $T$ is non-adaptive, the function queried at a certain height is independent of the earlier queries, let the function queried at height $i$ be $f_i$. As $T$ is non-adaptive it is also a complete binary tree. Represent each leaf of the binary tree by a string $s=s_1\dots s_k\in \{0, 1\}^k$. The string $s$ can be thought of as the results of the queries to $f_i$. That is, to reach a given leaf (represented by $s$) of $T$, the query functions $f_1, f_2, \dots f_k$ must evaluate to $s_1, s_2, \dots s_k$ respectively.

Consider any chain of inputs $\mathcal{X} = x^{(1)} \prec x^{(2)} \prec \dots x^{(\ell)}$ to $f$ and let the corresponding binary  respectively. We will argue that the alternation along this chain is bounded by $k$. Consider what happens when we move along the chain (say $x^{(j)}$ to $x^{(j+1)}$). Suppose that on inputs $x^{(j)}$ and $x^{(j+1)}$, $T$ evaluates to the leaves represented by the $k$ bit strings $s^{(j)}$ and $s^{(j+1)}$ respectively. Since $f_i$'s are monotone functions and $x^{(j)} \prec x^{(j+1)}$, we have $f_i(x^{(j)}) \leq f_i(x^{(j+1)})$ for all $i$. Thus, $s^{(j)} \prec s^{(j+1)}$ or $s^{(j)} = s^{(j+1)}$. As the value of $f$ is completely determined by the values of $f_i$'s, if $s^{(j)} = s^{(j+1)}$ then $f(x^{(j)})=f(x^{(j+1)})$. Hence, the number of alternations of $f$ along $\mathcal{X}$ is at most the number of strings $s^{(1)} \prec s^{(2)} \prec \dots$, which is at most the length of each string i.e., $k$.
\end{proof}

The above theorem along with Theorem~\ref{thm:mdt-alt} finishes the proof of Theorem~\ref{thm:intro-all-dtm} from the introduction.\\[-3mm]

\introthm{thm:intro-all-dtm}
{
For any Boolean function $f$, $\DTm(f) = \lceil \log(\mathsf{alt}(f) + 1) \rceil$, and $\DTmna(f) = \mathsf{alt}(f)$.
}
\vspace{-6mm}
\subsection{Monotone Certificate Complexity}		
We now discuss a characterization of non-adaptive monotone decision tree complexity through a generalization of certificate complexity of the function.
\begin{definition}[{\bf Monotone certificate complexity}]
For an input $x\in \{0,1\}^n$ of a Boolean function $f$, we call a set $S_x=\{f_1,f_2\dots f_k\}$ of monotone Boolean functions on $n$ variables as a {\em monotone certificate} (set) if for any input $y\in \{0,1\}^n$, we have that  $[\forall_{i=1}^k~ f_i(y)=f_i(x)]\Rightarrow [f(y) = f(x)]$. The monotone certificate complexity of $x$, denoted $\Cm(f,x)$ is defined as the minimum size $|S_x|$ of a monotone certificate $S_x$ of $x$. The monotone certificate complexity of the function $f$ itself is defined as $\Cm(f):=max_x \{\Cm(f,x)\}$.
\end{definition}

Interestingly, there is a constant upper bound of the size monotone certificate set for any function $f$. We show that, $\Cm(f)$ is at most $2$.

\begin{proposition}
\label{prop:ub-cmf}
$\Cm(f) \le 2$.
\end{proposition}
\begin{proof}
For any input $x$, let $I_x$ be the set of all variables set to 1 in $x$, and $J_x$ be the set of remaining indices. The set $S_x = \{f_1:=\bigwedge I_x, f_2:=\bigvee J_x \}$ is a valid monotone certificate for an input $x$. To show that this is indeed a certificate, we argue that $[(f_1(y)=f_1(x)) \wedge (f_2(y)=f_2(x))]\Rightarrow [f(y) = f(x)]$. First, note that $f_1(x)=1$ and $f_2(x)=0$ by definition of $f_i$'s. Now suppose the LHS of the above targeted implication is true. That is, $f_1(y)=1$ and $f_2(y)=0$. Hence, $\bigwedge I_x (y)=1$ and $\bigvee J_x (y)=0$, which means that $I_y \subseteq I_x$ and $J_y \subseteq J_x$. But note that $(I_y,J_y)$ (and $(I_x,J_x)$) is a partition of the total set of variables. Therefore, it must be the case that $I_x=I_y$ and $J_x=J_y$, meaning $x=y$. Then the RHS of the desired implication follows trivially.
\end{proof}

If the monotone certificate sets are constrained to be the same for all inputs, we call such a measure as the non-adaptive monotone certificate complexity of the function $f$, denoted by $\Cmna(f)$. By a simple argument, we have:
\begin{proposition}
\label{prop:cna-dtna}
$\Cmna(f)=\DTmna(f)=\alt(f)$.
\end{proposition}
\begin{proof}
\textbf{($\le $)} Let $h=\DTmna(f)$ denote the minimum height of a non-adaptive monotone decision tree (call it $T$) computing the Boolean function $f$ on $n$ variables. Let the query functions from root to any leaf be $f_1,f_2\dots f_h$ in order, where each $f_i$ is a monotone function on $n$ variables. We claim that $S:=\{f_1,f_2,\dots f_h\}$ is a monotone certificate for any input. Then immediately, $\Cmna(f)\le |S| = h =\DTmna(f) $. To observe that $S$ indeed is a certificate set, consider two inputs $x$ and $y$ for which all the functions in $S$ evaluate identically. Then in the tree $T$, both these inputs reach the same leaf, and hence return the same value. We obtain $[\forall_{i=1}^h~ f_i(y)=f_i(x)]\Rightarrow [f(y) = f(x)]$, and therefore $S$ is a monotone certificate for any input. 
		
\textbf{($\ge $)} Given that $S=\{f_1,f_2\dots f_k\}$ is a monotone certificate set, our goal is to design a non-adaptive monotone decision tree $T$ for $f$ of height $k$. The idea is to use the same functions that are in $S$ as queries in $T$. We fix an arbitrary ordering of $S$ and query the functions in the same order in all the paths. To get the labels of the leaves of $T$, we fix the label of a leaf $\ell$ as $f(x_\ell)$, where $x_\ell$ is any arbitrary input that reaches $\ell$ in its computation by $T$. If no such $x_\ell$ exists, then $\ell$ may be labeled arbitrarily. Let some input $x$ reach a leaf $\ell$ in $T$. As $x_\ell$ also reaches $\ell$, both $x$ and $x_\ell$ must evaluate identically on all the functions in $S$. By the certificate property of $S$, this implies that $f(x)=f(x_\ell)$, which is exactly what is returned by $T$ on input $x$.
\end{proof}  

\section{Non-deterministic Monotone Decision Trees}
\label{sec:ndt}
Inspired by the definitions of a non-deterministic decision tree and certificate complexity of a Boolean function, we study a non-deterministic variants of monotone decision trees as well. We define
a non-deterministic monotone decision tree as a tree where there can be single or multiple outgoing edges at each internal node, and each edge in the tree is labeled by a monotone function or the negation of a monotone function, and the leaves are labeled 0 or 1. An input is said to be accepted if there is at least one path from the root to a leaf labeled $1$ along which all the functions appearing as labels on the edges evaluate to $1$. 

\subsection{Two equivalent definitions} We define two variants of a non-deterministic monotone decision tree. The first variant is the one that we will use in the later part of this section. 
\begin{itemize}
	\item {\sf Model1}: A tree where there can be any number of outgoing edges at each internal node, and each edge is labeled by a monotone or negation of a monotone function. The leaves are labeled with 0's and 1's. 
	
	An input $x$ to the tree is said to be accepted (same as outputting 1 on $x$) iff there exists a path from the root of the tree to any leaf labeled 1 along which all the edges are \emph{active}. Here we say an edge labeled $f_i$ is \emph{active} on $x$ when $f_i(x)=1$.
	\item {\sf Model2}: A tree where there can be any number of outgoing edges at each internal node, and each edge is labeled 0 or 1. The leaves are labeled with 0's and 1's as usual. In addition to these labels, the internal nodes are also labeled, with monotone functions. 
	
	An input $x$ to the tree is said to be accepted (same as outputting 1 on $x$) iff there exists a path from root of the tree to  any leaf labeled 1 along which all the edges are \emph{active}. Here, for an edge $e$ outgoing from a node labeled $f_i$, we say $e$ is \emph{active} on input $x$, if the label of $e$ is equal to the value $f_i(x)$.
\end{itemize}

We say a tree $T$ (could be of type {\sf Model1} or {\sf Model2}) is said to compute a Boolean function $f:\{0,1\}^n\to \{0,1\}$ iff for all $x\in \{0,1\}^n$, $T$ outputs $f(x)$ on $x$.

The above defined models can be shown to equivalent. That is, we show that if a Boolean function $f:\{0,1\}^n\to \{0,1\}$ can be computed by a tree of type {\sf Model2}, then there is also a tree of type {\sf Model1} of same height (upto a constant factor) computing $f$; and vice-versa. 

To prove the first part, let $T$ be a {\sf Model2} tree computing $f$. We simply re-label $T$ as follows, to obtain a {\sf Model1} tree still computing $f$. At each internal node labeled $f_i$; if an outgoing edge is labeled 0, replace the label with $\overline{f_i}$, if an outgoing edge is labeled 1, replace the label with $f_i$. Finally, remove all the labelings at the internal nodes. Retain the labels of the leaves. Thus, we get a tree $T'$. Note that an edge in $T$ is active if and only if the corresponding edge in $T'$ is active for the same input. Therefore, any path that is active in $T$ remains active in $T'$, meaning $T\equiv T' \equiv f$.

The other direction is proved as the following lemma.

\begin{lemma}
	\label{lem:model-12}
	If a Boolean function $f:\{0,1\}^n \to \{0,1\}$ is computed by a tree $T$  of type {\sf Model1}, there is also a tree $T'$ of type {\sf Model2} of height twice that of $T$, computing $f$.
\end{lemma}

\begin{proof}
	To construct $T'$, we perform the following operation over all the edges of $T$ from the top to bottom. Let $e$ be an outgoing edge in $T$, from node $i$ labeled $f_i$ (could be a monotone or negation of a monotone function) to a node/leaf called $j$. A new node called $k$ is introduced in $T'$ between the nodes $i$ and $j$, and the edges $i\to k$ and $k\to j$ are added. The label at the node $i$ is changed to the constant function \textbf{1} from $f_i$, the $i\to k$ edge is labeled 1, and $k$ is labeled with the function $f_i$. If $f_i$ is monotone, $k\to j$ edge is labeled 1, Otherwise, $k\to j$ edge is labeled 0. By this labeling principle, notice that for any input $x$, the edge $i\to k$ is always active and the edge $k\to j$ is active if and only if $x$ activates the edge $e$ in $T$. Therefore, after the complete construction of $T'$, there is an active path to a leaf if and only if there is an active path to that same (corresponding) leaf in $T$. Since, the existence of active path to a leaf with label 1 is the criterion of acceptance in either tree, both trees output the same value for the same input, and hence $T'$ computes $f$. Finally, we notice that the height of $T'$ is twice the height of $T$, as each edge in $T$ produces two edges in $T'$ by our design.
\end{proof}

\subsection{Characterization using Alternation}

Analogous to the normal monotone decision trees, we prove a bound on the height and the size as the complexity measures of the non-deterministic monotone decision tree (height denoted by $\DT_m^n(f)$) in the following Theorem.

\begin{theorem}
\label{prop:nmdt}
For any Boolean function $f$, $\DT_m^n(f) \le 2$ and size of the optimal non-deterministic monotone decision tree is $\lceil{\frac{{\sf  alt}(f)}{2}}\rceil$.
\end{theorem}
\begin{proof}
We know that any Boolean function $f$ has a monotone decomposition of length at most $k:={\sf alt}(f)+1$ (Lemma~\ref{lem:alternation-decomp}), so let $f=f_1 \oplus f_2 \oplus \dots \oplus f_k$. Due to the implication property of $f_i$'s, we have the following equivalent evaluation for $f$, that is $f= \overline{f_1} f_2 \vee \overline{f_3} f_4 \vee \dots \overline{f_{k-1}} f_k$  \footnote{If $k$ is odd, we can just make $f_1$ as the constant function $\textbf{0}$.} (Proposition~\ref{prop:decomp-dl}). We now construct the non-deterministic monotone decision tree $T$ of height 2 size as much as the number of ``terms'' in the above representation of $f$, which is $s:=\lceil k/2 \rceil$. The root of $T$ would contain $s$ edges with labels $\overline{f_1},\overline{f_3},\dots ,\overline{f_{k-1}}$, each for one edge. Then for each of these children of root, we give one child with edge label $f_{i+1}$, corresponding to $\overline{f_i}$. Finally we attach a leaf labeled 1 to all the leaves of the tree. There is an active path from root to a leaf if and only if $f$ evaluates to 1 on that input, thereby showing that of $T$ computes $f$. 
\end{proof}

\section{Randomized Monotone Decision Trees}
\label{sec:rdt}
We also study randomized monotone decision trees. In this model, monotone query nodes in the decision tree, random bit choices are also allowed at the internal nodes of the tree and each of the random choice nodes also has two outgoing edges to children one with labeled $0$ and the other labeled $1$. We say the tree computes a Boolean function $f$ if for any input $x$, the probability (over the choice of the settings for the random bit choices in the tree) of the computation reaching a leaf with label $f(x)$ is at least 
$2/3$. By $\DTmr(f)$, we denote the minimum height of a RMDT computing a Boolean function $f$. The following theorem implies that randomization does not help when the monotone queries are unrestricted.

\begin{theorem}
\label{thm:rmdt-main}
For any Boolean function $f$, $\DTmr(f) = \DTm(f) = \lceil\log(\alt(f)+1) \rceil$.
\end{theorem}
\begin{proof}
	Since any monotone decision tree is trivially also a RMDT, we directly have $\DTmr(f) \le \DTm(f)$. To prove the reverse inequality, we will construct a circuit for $f$ with at most $\DTmr(f)-1$ negation gates. Then by using Theorem~\ref{prop:anadt-neg}, we get a (deterministic) MDT of height $\DTmr(f)$ computing $f$.
	
	Let $T$ be a RMDT of minimum height $h=\DTmr(f)$, computing $f$. Let the number of leaves in $T$ labeled with 1 and 0 be $k$ and $k'$ correspondingly. We have $k+k'=\text{(total no. of leaves in $T$)}\le 2^h$. We will give the proof only when $k\le k'$. Otherwise, we can consider the RMDT obtained by flipping all the leaf labels of $T$, which would then compute $\overline{f}$. As $\DTm(f)=\DTm(\overline{f})$ and $\DTmr(f)=\DTmr(\overline{f})$\footnote{On flipping the leaf labels in a MDT or RMDT, it computes the complement function.}, the required result follows. So, we may assume $k\le k'$.
	
	We will first derive a characterization for the function $f$, which will be helpful in other proofs too.
	
	Consider a fixed input $x$ to $T$. By the definition of RMDT, $f(x)=1$ iff the probability of the computation of $T$ (on $x$) reaching a 1-labeled leaf is at least half \footnote{Although a stronger bound of 2/3 is implied, the bound of 1/2 will be simpler to work with.}. Let us now express this probability in terms of the structure of $T$. 
	
	Let $\ell_1,\ell_2,\dots \ell_k$ be all the 1-labeled leaves, $r_1,r_2,\dots r_k$ be the number of random nodes from root to the corresponding leaf; and let $c_i:=(\bigwedge f_{pi}) \wedge(\bigwedge \overline{f_{qi}})$ denote the \emph{characteristic function} corresponding to $\ell_i$, where the $f_{pi}$'s are the monotone queries to be passed and $f_{qi}$'s to be failed in the root to $\ell_i$ path. First we calculate the probability that a specific leaf $\ell_i$ is reached upon the computation. Once we get this value, since in any given circumstance, the computation reaches a unique leaf, the above events for various $i$'s are mutually exclusive, meaning the desired probability is simply the sum of probabilities that the computation reaches a particular $\ell_i$. 
	
	Now, coming to the probability of reaching a particular $\ell_i$, it is zero when $c_i(x)=0$. Indeed when $c_i(x)=0$, by its definition, we can observe that it means at least of the monotone functions that is supposed to pass has failed or at least one that is supposed to fail has passed, either of which is a contradiction. Now for the case $c_i(x)=1$, the probability solely depends on $r_i$: As all the intermediate monotone queries support the computation to reach $\ell_i$ (i.e, $c_i(x)=1$), for any of the random nodes in the path, there is exactly one result that will keep the computation in the right path to $\ell_i$. Hence the probability then is equal to $(1/2)^{r_i}$. The probabilities in both these cases can be compacted as $(1/2)^{r_i}c_i(x)$. We can see that it is 0 when $c_i(x)=0$, and $(1/2)^{r_i}$ when it is equal to 1. 
	
	Going back to the overall probability, it is equal to the sum $\sum_i (1/2)^{r_i}c_i(x)$. We thus have the characterization: 
	\begin{align}f(x)=1 \text{~~iff~~} \sum_{i=1}^k 2^{h-r_i}c_i(x)\ge 2^{h-1}.
		\label{eqn:1}
	\end{align}
	
	Notice that to compute any $c_i$, at most one negation is needed, since $c_i=(\bigwedge f_{pi}) \wedge (\neg(\bigvee f_{qi}))$.
	Now, to construct a circuit for $f$ with $h-1$ negations as stated in the beginning, we observe that: since the threshold function is monotone, if we can obtain all the $c_i$'s using a (multi-output) circuit using at most $h-1$ negations, we are done. We make use of Fischer's construction \cite{Fis75} for this, where a multi-output circuit using $\lceil\log(m+1) \rceil$ negations is designed to compute the inverting function $I(z_1,z_2,\dots ,z_m)=(\overline{z_1},\overline{z_2}\dots,\overline{z_m})$. Taking $m=k$ and $z_i=\bigvee f_{qi}$; the number of negations used in the circuit would be ${\sf neg}:=\lceil \log(m+1) \rceil=\lceil \log(k+1) \rceil$. 
	
	If $k<k'$, then ${\sf neg}=\lceil \log(k+1) \rceil \le \lceil \log(2^{h-1}-1+1) \rceil=h-1$. 
	
	Otherwise, $k=k'=2^{h-1}$. Suppose the right-most leaf in $T$ is labeled 1. Then note that there is no negation in $c_k$ and so, we do not have to negate $z_k$; making the number of inputs used in the Fischer's construction only $k-1=2^{h-1}-1$. Hence, ${\sf neg}=\lceil \log(2^{h-1}-1+1) \rceil=h-1$. Now, if the right-most leaf in $T$ is instead labeled 0, then we flip all the leaf labels in $T$ to obtain a RMDT for $\overline{f}$, which then falls into the above case (of the right-most label being 1).
	
	By Theorem~\ref{prop:anadt-neg}, we can now using this circuit, obtain an MDT of height $h-1+1=h=\DTmr(f)$ computing $f$. Hence, $\DTmr(f)=\DTm(f)$. 
\end{proof}

\noindent{ \bf A variant of Randomized Monotone Decision Tree Model:}
We also study a more powerful variant of the randomized model where each node is allowed to have a multi-set of $w$ monotone functions associated with it (which we call the query set) and on an input $x$ to the decision tree, at each node, one of the query functions is chosen uniformly at random from the corresponding query set. Again, we say that the tree computes a Boolean function $f$ if for any input $x$, the probability of the computation reaching a leaf with label $f(x)$ is at least $2/3$. We denote by $\DTmrbag{w}(f)$, the minimum height of such a randomized decision tree that computes $f$. It can be observed that any RMDT can be implemented in this model as well (query sets are of size $w$): a monotone query $f_i$ can be replaced with the query set $\{f_i,\dots,f_i (w \text{ times})\}$, and a node with a random bit choice by the query set $\{\textbf{0},\dots,\textbf{0}(w/2 \text{ times}),\textbf{1},\dots,\textbf{1}(w/2 \text{ times})\}$. This gives $\DTmrbag{w}(f) \le \DTmr(f)$ for any even $w \ge 2$. Even if $w$ is odd, the relation still holds up to a constant factor. For the other direction, we show the following:

\begin{theorem}
\label{thm:w-RMDT}
For any Boolean function $f$, $\DTmr(f)\le (1+k).\DTmrbag{w}(f)$, if $w=2^k$ for an integer $k$.
\end{theorem}

\begin{proof}
Let $T$ be the tree corresponding to the stronger model with query sets of size $w=2^k$ of minimum height computing $f$. To transform $T$ into an RMDT, we perform the following transformation at each internal node from bottom to top. 

Let $u$ be an internal node of $T$ with left branch going to the sub-tree $u_L$, right branch to $u_R$, and with possible query function at $u$ being from $\{q_1,q_2,\dots q_w\}$. Consider a complete binary tree $B$ with nodes in the first $k$ levels making random queries to 0 or 1, and the $2^k=w$ nodes in the bottom-most level making corresponding queries to the $q_i$'s. 

We replace the subtree of $T$ rooted at $u$ with the following tree: $B$ is introduced at the top (in place of $u$) and each of the bottom-most nodes of $B$ are branched into (copies of) the sub-trees $u_L$ on left and $u_R$ on right. Clearly, upon applying this procedure for all the internal nodes of $T$ bottom-up, we get an RMDT. 

To justify that this transformation still computes $f$, it is sufficient to argue that for an input $x$, due to the transformation corresponding to a node $u$, the probabilities that $u_L$ and $u_R$ are taken by the computation remain the same after the transformation. This is indeed true as these quantities are equal to $|\{q_i~|~q_i(x)=0\}|/w$ and $|\{q_i~|~q_i(x)=1\}|/w$ respectively, before and after the transformation. 

Since each original node is replaced with a complete binary tree $B$, the height increases by a factor equal to the height of $B$ (including bottom-most level), i.e, $1+k$. Hence, $\DTmr(f)\le (1+k).\DTmrbag{w}(f)$.
\end{proof}

\section{Monotone Decision Trees with Query Restrictions}
\label{sec:mdt-query-restrict}

In this section, we study the power of monotone decision trees under restricted monotone query functions. We define $\MDT{\calC}$ as the set of Boolean functions (rather Boolean function families) that admit decision trees of height $\bigO(\log n)$, where $n$ is the number of variables and all the query functions are from $\textit{mon-}\calC$. We could instead consider polynomial size decision trees, but both formulations turn out to be equivalent for interesting classes $\calC$. Similarly, we define $\MDL{\calC}$ as the set functions computed by MDLs of polynomial size where the queries are in $\textit{mon-}\calC$.

We first justify our reason to consider the height bound as $\bigO(\log n)$. We show that for any $h=\omega(\log n)$, there is a function $f$ on $n$ variables that has a decision tree of height $h$ with query functions computed by monotone polynomial sized circuits, but $f$ cannot be computed by a polynomial size circuit. In contrast, for any $h= o(\log n)$, if a function $f \in \calC$ on $n$ variables has alternation $\Omega(n)$, then $f$ does not have a  decision tree of height $h$, with query functions computable by monotone circuits in $\calC$. Hence, the question of whether $\MDT{\calC}$ is equal to $\calC$ is well-motivated only when the height is $\bigO(\log n)$. With this background, we will then study $\MDT{\calC}$ as defined above.
	
\subsection{Height Constraints on $\MDT{\calC}$}

\begin{proposition}
\label{prop:omega-logn}
For any $h=\omega(\log n)$, there is a Boolean function $f$ on $n$ variables that has a decision tree of height $h$ with query functions computed by monotone polynomial sized circuits, but $f$ cannot be computed by polynomial size circuits. 
\end{proposition}
\begin{proof}
We will actually show the existence of a function that has a `simple decision tree' (all queries are to input variables) of height $h=\omega(\log n)$ but not any polynomial sized circuit. By Shannon's counting argument (see \cite{Juk12}) we know that there exists a function on $h$ variables which requires a circuit of size $\Omega(2^h/h)$ to compute it. Our function $f$ would be the same function but with an additional $(n-h)$ many dummy variables. Since $f$ depends on only $h$ variables, we can obtain a (non-adaptive) decision tree where the queries are the bits that the function depends on. Its height clearly is $h$. By definition, any circuit computing $f$ has size $\Omega(2^h/h) \ge c_1.(2^h/h)$. \footnote{$c_1,c_2,c_3$ are some fixed positive constants and the inequalities are asymptotic i.e, for sufficiently large $n$.}
		
For the sake of contradiction, suppose that $f$ does have a polynomial sized circuit. It means $c_1.(2^h/h) \le n^{c_2}$. Taking logarithm, we get $c_3+h-\log h \le c_2 \log n$, so we have, $h \le 2h-2\log h\le 2.(c_3+h-\log h)\le 2c_2 \log n=\mathcal{O}(\log n)$, which contradicts $h=\omega(\log n)$. Therefore, $f$ cannot be computed by a polynomial size circuit family.
\end{proof}
		
\begin{proposition}
\label{prop:o-logn}
Let $\calC$ be any circuit complexity class which contains a function $f$ with $\alt(f) = \Omega(n)$. For any $h= o(\log n)$, there is a function $f \in \calC$ on $n$ variables that does not have a  decision tree of height $h$, with query functions computable by monotone circuits in $\calC$. 
\end{proposition}
\begin{proof}
Let $f\in \calC$ be a function such that $\alt(f)=\Omega(n)$.
For the sake of contradiction, assume that $f$ does have a monotone decision tree of height $h=o(\log n)$. We know that $\alt(f)\le 2^{h}$ where $h$ is the height of the decision tree (recall Theorem~\ref{thm:mdt-alt}). Then, $\alt(f) \le 2^{o(\log n)} = o(n)$, contradicting $\alt(f)=\Omega(n)$. 
\end{proof}

\vspace{1mm}

\subsection{Deterministic MDTs with Query Restrictions: $\MDT{\calC}$ vs $\calC$}
As mentioned in the introduction, we ask: How much can monotone decision tree computation, with query functions computable by monotone circuits in the class $\calC$, simulate general computation in the class $\calC$. In this direction, we first show that $\MDT{\calC} \subseteq \calC$ when $\calC$ has reasonable closure properties.

\begin{lemma}
\label{lem:mdtc-c}
For a circuit complexity class $\calC$ closed
under polynomially many $\lnot, \land, \lor$ operations, $\MDT{\calC} \subseteq \calC$.
\end{lemma}

\begin{proof}
The proof of Theorem~\ref{thm:mdt-mdl} also establishes that $\MDT{\calC}=\MDL{\calC}$ as log-height MDTs and poly-size MDLs are inter-convertible, it suffices to show that $\MDL{\calC} \subseteq \calC$. Let a Boolean function $f$ belong to $\MDL{\calC}$ via the decision list $L=(f_{1},c_1)(f_{2},c_2)\dots (f_{k},c_k)$ where $k=\poly(n)$; and each query function $f_{i}$ has a (monotone) circuit $C_i$ from the class $\calC$. Using the normal form for the decision lists for circuit classes with the above property, we will assume that the $c_i$'s are alternating; and the query functions $f_i$ are forward firing, i.e $f_1 \Rightarrow f_2 \Rightarrow \dots \Rightarrow f_{k}$. As we can always prepend a $(\textbf{0},0)$ node at the beginning or append a $(\textbf{1},1)$ node at the end of $L$ while still maintaining the normal form; w.l.o.g we may assume that $c_1=0$ and $c_k=1$ (hence $k$ is even). Due to the alternating constants property, this means $c_{2i}=1$ and $c_{2i-1}=0$. 

We know from Proposition~\ref{prop:decomp-dl} that the Boolean function $g:=\overline{f_1}f_2 \vee \overline{f_3}f_4 \dots \vee \overline{f_{k-1}} f_k$ is equivalent to $f$. Finally, we note that because of the closure properties of $\calC$ and $k=\poly(n)$, the expression $g$ can be used to obtain a circuit in $\calC$ computing $f$.
\end{proof}

If the class $\calC$ is rich enough to include monotone circuits for the threshold functions, for example say the class $\TC^0$ itself, then we can actually prove equality: Note that the Monotone Decomposition given in Lemma~\ref{lem:tc0-decomp} can be easily transformed into a MDL with the same functions being queries using Proposition~\ref{prop:decomp-dl}. Thus, we get $\calC \subseteq \MDL{\calC}$, which when combined with the fact that $\MDT{\calC}=\MDL{\calC}$ and Lemma~\ref{lem:mdtc-c} completes the proof of Theorem~\ref{introthm:mdtc-c}.

\introthm{introthm:mdtc-c}
{
	For any circuit complexity class $\calC$ such that $\textit{mon-}\TC^0 \subseteq \textit{mon-}\calC$, 
	$\MDT{\calC} = \calC$.
}

\subsection{Monotone Decision Trees and $\AC^0$}	
\label{subsec:mdt-ac0}
We attempt to address the question $\MDT{\AC^0}$ vs $\AC^0$. We know that $\MDT{\AC^0} = \MDL{\AC^0}$ is contained in $\AC^0$ by Lemma~\ref{lem:mdtc-c}. An interesting challenge is to prove or disprove the reverse containment. As a warm-up, we show that $\MDT{\AC^0}$ is more powerful than polynomial sized term decision lists (which is a strict subset of $\AC^0$).

\begin{proposition}
\label{prop:MDT-TDL}
$\MDT{\AC^0} \nsubseteq \mathsf{TDL}$.
\end{proposition}
\begin{proof}
We show that $\MDL{\AC^0} \nsubseteq \mathsf{TDL}$. 
Using the construction in Lemma~\ref{lem:mdtc-c} we know that all functions in $\mathsf{TDL}$ have polynomial sized circuits of depth $4$ since the query functions are depth $2$ monotone circuits. So, if we can show the existence of a function $f$ that is in $\MDL{\AC^0}$ but has no polynomial sized circuit of depth 4, we are done. It is known for any depth $d$, there is a monotone function $f$ which can be computed by a monotone circuit of depth $d$ but cannot be computed by any polynomial size (even non-monotone) circuits of depth $d-1$ (See \cite{Sip83}). As the decision list $L=(f,1)(1,0)$ computes $f$ and the query function $f$ has a monotone $\AC^0$ circuit, we have $f\in \MDL{\AC^0}$. And $f \notin \mathsf{TDL}$, because all functions in $\mathsf{TDL}$ have $\AC^0$ circuits of depth 4 and $f$ by definition has none.
\end{proof}
%
%
%


Moving towards comparing the class with $\AC^0$,
we first apply Propositions~\ref{prop:o-logn}\footnote{To apply Prop.~\ref{prop:o-logn}, take  $f(x)=x_1\overline{x_2} \vee x_3\overline{x_4} \dots \vee x_{n-1}\overline{x_n}$.} and~\ref{prop:omega-logn} to $\AC^0$, and conclude that:
For any $g(n) = o(\log n)$, and $h(n) = \omega(\log n)$,~
$\MDTTT{g(n)}{\AC^0} \subsetneq \AC^0$ and $\MDTTT{h(n)}{\AC^0} \nsubseteq \AC^0$.
In contrast to this, we show that the whole of $\AC^0$ can be computed by monotone decision trees with some sub-linear height. By using a theorem from due to Santha and Wilson~(See Theorem~4.1 of \cite{SW93}), which reduces the number of negations in the circuit to $\frac{n}{\log^r n}$, and then applying Theorem~\ref{prop:anadt-neg}, we show:

\begin{theorem}
\label{thm:sublinear-ac0}
For any constant $r$, $\AC^0 \subseteq \MDTTT{d(n)}{\AC^0}$ where $d(n) = {\Omega}\left(\frac{n}{\log^r n}\right)$.
\end{theorem}
\begin{proof}
We use the following theorem from due to Santha and Wilson~(See Theorem~4.1 of \cite{SW93}): For any constant $r$, every $\AC^0$ circuit family can be translated to a constant depth polynomial size circuit that uses at most $\bigO\left(\frac{n}{\log^r n}\right)$ negations. We can directly combine this with Theorem~\ref{prop:anadt-neg} to get a decision tree of height at most $\bigO\left(\frac{n}{\log^r n}\right)$ by observing that since the queries used in the proof of Theorem~\ref{prop:anadt-neg} are monotone sub-circuits\footnote{By sub-circuit, we mean the circuit obtained by fixing the values of some intermediate gates to 0 or 1.} of the original circuit, in this case, they are computable in monotone $\AC^0$.
\end{proof}

\noindent 
{\bf A negation-limited computation of $\MDT{\AC^0}$:} We show that the functions in $\MDT{\AC^0}$ have $\AC^0$ circuits with `limited' negation gates.


\introthm{introthm:ne-ac0}
{ 
If a Boolean function $f$ on $n$ variables is in $\MDT{\AC^0}$, then for any positive constant $\epsilon\le 1$, there is an $\AC^0$ circuit for $f$ with $\bigO(n^\epsilon)$ negation gates.
}
\begin{proof}
As $f\in \MDL{\AC^0}$, by assuming that the MDL is in normal form and applying Proposition~\ref{prop:decomp-dl}, we can write $f=\overline{f_1}f_2 \vee \overline{f_3}f_4 \vee \dots \overline{f_{\ell-1}}f_\ell$, where $\ell= \bigO(n^k)$ for some constant $k$. In addition, all the $f_i$'s have monotone $\AC^0$ circuits and $\forall i \in [\ell-1], f_i \Rightarrow f_{i+1}$.
Thus, it suffices to produce $\overline{f_i}$ for every $i \in [\ell]$ which is odd, from $f_1, \ldots f_n$, using a constant depth polynomial size circuit that uses $\bigO(n^\epsilon)$ negations. Indeed, the trivial circuit uses $\ell = \bigO(n^k)$ negations.
%
		 
The main observation is that the bits (the outputs of $f_i$ where $i$ is odd) we need to invert are already in sorted order, since $\forall i, f_i \Rightarrow f_{i+1}$.
Let this bit-string be $s = 0^j1^{m-j}$, where $m:=\lceil \ell /2 \rceil$. We need to output $\overline{s} =1^j0^{m-j}$. 

As proved in \cite{SW93} (Theorem 3.6 in \cite{SW93}), this can be implemented using an iterative construction which uses only $\bigO(n^{\epsilon})$ negations. In the proof of Theorem 3.6 in \cite{SW93}, the authors also observe, this part of their construction uses only polynomially many $\neg$, and (unbounded) $\land$ and $\lor$ gates. Hence the final circuit is within $\AC^0$.
We present the construction in our context in Appendix~\ref{appsec:Inverter} for reference.
\begin{aproof}{$\AC^0$ Inverter Construction for Sorted Inputs - Adaptation from \cite{SW93}}{appsec:Inverter}

We elaborate on the construction used in the proof of Theorem~\ref{introthm:ne-ac0}.
We present the construction due to \cite{SW93} in our simpler notation and setting. Divide $s$ into $t=n^\epsilon$ contiguous blocks $B_1,B_2,\ldots ,B_t$ each of length $p:= m/n^\epsilon = \bigO(n^{k-\epsilon})$. Observe that the negation of  block $B_i$ is of the form $1^p$ or $0^p$ or $1^j0^{p-j}$ for some $0 \le j \le p$, based on whether $B_i$ witnesses switching from 0s to 1s in $s$. Our objective is to construct an $\AC^0$ circuit using only $\bigO(n^{\epsilon})$ negations.

Let $B_r$ denote the block which contains this index where the switch happens. Call such a block \textit{special} for a given $s$. Note that there is at most one such block. For each block $B_i$, define $b_i = B_i[1] \oplus B_i[p]$ -- i.e., the parity of the first and last bits of the block $B_i$. Notice that the bit $b_i$ indicates whether $B_i$ is special or not. Thus, we can obtain the special block as follows $B_r = \lor_{i=1}^t \left((b_i)^p \land B_i\right)$ where $\land$ and $\lor$ are bit-wise, and $(b_i)^p$ is the bit $b_i$ repeated $p$ times. 

Once we have the bits of the special block $B_r$, we naively invert all its individual bits carried in the wires produced above to get $B_{neg}:=\overline{B_r}$. What remains is to identify which type each block $B_i$ is, and appropriately wire to produce $1^p$ or $0^p$ or $B_{neg}$ as their inversions.

This is done as follows: the inversion of $B_i$ is $(\left(b_i\right)^p \land B_{neg}) \lor  (\left(\overline{b_i}\right)^p \land (\overline{c_i})^p)$ where $c_i$ is the first bit of $B_i$. To check this, note that the output of the above expression is $B_{neg}$ if $b_i=1$; otherwise it is $0^p$ if $c_i=1$, and is $1^p$ if $c_i=0$. The number of negations used is $\bigO(n^\epsilon)+\bigO(n^{k-\epsilon})$, the first part for producing $\overline{b_i}$'s and $\overline{c_i}$'s, and the second part to produce $B_{neg}$, all of which can be wired commonly for all the blocks. If $\epsilon > \frac{k}{2}$, we are done. Otherwise, we have reduced the number of negations to $\bigO(n^{k-\epsilon})$. To reduce the exponent even further, we apply this construction recursively to invert $B_r$, which again, by definition sorted. This reduces the exponent additively by $\epsilon$ by suffering an increase of only constant depth at each level. Thus, in $\lceil k/\epsilon \rceil$ levels we would get an inverter with only $\bigO(n^{\epsilon})$ negations. The size of the circuit is polynomial in $n$ as the number of gates introduced in each level is only linear in the number of sorted bits.
\end{aproof}
\end{proof}
In contrast, we note that certain $\AC^0$ circuits require a lot of negations.
\begin{theorem}[Theorem 3.2 of \cite{SW93}]
\label{thm:neg-ac0-lb}
For every $f \in \AC^0$ with $\alt(f) = \Omega(n)$, and for every $\epsilon > 0$,
any $\AC^0$ circuit computing $f$ will have at least $\Omega(n^\epsilon)$ negation gates for some positive constant $\epsilon$ (that can depend on the circuit). 
\end{theorem}
\begin{proof}
This follows directly from \cite{SW93} where the authors show that\footnote{Although this was stated for multi-output functions in~\cite{SW93}, it holds for single-output functions as well.}, if any Boolean function $f$ of alternation $k$ is computed by a circuit $C$  of depth $d$, then the number of negations in $C$ is at least $d(k+1)^{1/d} - d$. This is $\Omega(n^\epsilon)$ as $d$ is constant and $k=\Omega(n)$.
\end{proof}

Thus, if Theorem~\ref{thm:neg-ac0-lb} can be improved asymptotically for some $f \in \AC^0$ and fixed $\epsilon$, then we can show that $\MDT{\AC^0}$ is strictly contained in $\AC^0$.\\

\noindent {\bf A candidate function for $\MDT{\AC^0}$ vs $\AC^0$ question:}
We now show that there is a simple function that can be computed by depth two $\AC^0$ circuits, which if shown to be in $\MDT{\AC^0}$ will imply that $\MDT{\AC^0}=\AC^0$. This in particular gives a potential candidate function for the separation of the two classes.

\begin{theorem}
    If the family of functions $f^{n}(x)=\overline{x_1}x_2 \vee \overline{x_3}x_4 \vee \dots \overline{x_{n-1}}x_n$ is in $\MDT{\AC^0}$, then $\MDT{\AC^0}=\AC^0$.
\end{theorem}
\begin{proof}
Suppose that the function $f^n$ is in $\MDT{\AC^0} = \MDL{\AC^0}$. Assuming the MDL is in normal form, by Proposition~\ref{prop:decomp-dl}, $f^n$ can be expressed in various forms as follows:
\begin{align}
	f^n & = (f_1^n,0)(f_2^n,1)\dots (f_{p(n)}^n,1)(f_{p(n)+1}^n,0) \nonumber\\
	& = f_1^n \oplus \dots \oplus f_{p(n)+1}^n\nonumber\\
	& =\overline{f_1^n}f_2^n \vee \overline{f_3^n}f_4^n \vee \dots \vee \overline{f_{p(n)-1}^n}f_{p(n)}^n \nonumber\\
	& =({f_2^n} \vee \overline{f_3^n}) \wedge ({f_4^n} \vee \overline {f_5^n}) \wedge \dots \wedge ({f_{p(n)}^n} \vee \overline{f_{p(n)+1}^n}),
	\label{eqn:2}
\end{align}
where for all $i$, $f_i^n \in $ {\em mon-}$\AC^0$, and $f_i^n \Rightarrow f_{i+1}^n$, and $p(n)$ is even and some polynomial in $n$. Suppose $p(n)=n^{c_1}$ and each $f_i^n$ has a monotone $\AC^0$ circuit of depth $d$ and size at most $n^{c_1}$ for some constants $d \ge 1,c_1 \ge 2$.
    
Now we will show that $\AC^0 \subseteq \MDT{\AC^0}$. Consider an arbitrary function $g \in \AC^0$ computed by a De-Morgan formula\footnote{$\AC^0$ is equivalent to polynomial size constant depth De-Morgan formulas i.e., Boolean formulas in which the negations are only at input variables. W.l.o.g., we also assume that the AND and OR gates are in alternate levels of the formula.} $G$ of depth at most $e$ and size at most $n^{c_2}$ for some constants $e$ and $c_2\ge 2$. 
We will show that $g$ has a MDL $L_g$ of size (in its normal form) at most $n^{(c_1c_2)^{2e}}$ in which all the queries have monotone circuits of depth at most $de$ and size at most $n^{(c_1c_2)^{2e}}$. Then, since $e,c_1,c_2$ are constants, the size of the monotone circuits and the number of queries are polynomial in $n$, so $g \in \MDL{\AC^0}$. 

We show the existence of $L_g$ by induction on $e$.
    
\noindent {\bf Base case: $e = 1$.} Depth-1 formulas are just OR or AND of literals. If $g=(\bigvee_i x_i) \vee (\bigvee_j \overline{x_j})$ where $x_i$'s and $x_j$'s are variables, then $L_g = (\bigvee_i x_i, 1)(\bigwedge_j x_j, 0)(\textbf{1},1)$ works. The case of OR gate is similarly handled. The depth, size, and number of queries in $L_g$ are therefore bounded as expected (assuming $n$ is not too small).

    
\noindent {\bf Induction step: $e \ge 2$.} Again, we assume that the root gate is an OR gate: $g = \bigvee_{i=1}^s h_i$, where each $h_i$ has an $\AC^0$ formula of depth at most $e-1$ and size at most $n^{c_2}$, which means by the induction hypothesis, that it has an MDL $L_{h_i}$ of size at most $n^{(c_1c_2)^{2e-2}}$ such that all its queries have monotone circuits of depth at most $d(e-1)$ and size at most $t:= n^{(c_1c_2)^{2e-2}}$ -- let us say that $L_{h_i} = \overline{f_{i,1}}f_{i,2} \vee \overline{f_{i,3}}f_{i,4}\vee \dots \vee\overline{f_{i,t-1}}f_{i,t}$. Then, we have $g = \bigvee_{i=1}^s h_i= \bigvee_{i=1}^s  (\overline{f_{i,1}}f_{i,2} \vee \overline{f_{i,3}}f_{i,4}\vee \dots \vee \overline{f_{i,t-1}}f_{i,t})$. The trick now is to notice that this expression for $g$ looks exactly like the Boolean function $f^n$, except the variables are replaced by some monotone functions $f_{i,j}$'s, each of which has monotone circuits of depth at most $d(e-1)$ and size at most $t$. That is, $g = f^{st}(f_{1,1},\dots ,f_{1,t}, f_{2,1},\dots ,f_{2,t}, \dots , f_{s,1},\dots ,f_{s,t} )$. Now, using the MDL (family) we have for $f$ from Equation~\eqref{eqn:2}, we get an MDL $L_g$ for $g$ by substituting the variables $x_i$'s with functions $f_{i,j}$'s. The number of queries is $p(st)=(st)^{c_1}$, each query has a monotone formula of depth at most $d+d(e-1)=de$ and size at most $(st)^{c_1} + st.t \le (st)^{c_1 + 2} \le (n^{c_2}n^{(c_1c_2)^{2e-2}})^{c_1+2} \le n^{(c_2 + (c_1c_2)^{2e-2})(c_1+2)} \le n^{(c_1c_2)^{2e}}$ (using $e,c_1,c_2 \ge 2$).

    
    A similar construction works if the root gate is an AND gate, in which case we can make use of the ``CNF form'' for $f$ instead of the ``DNF form''. 
  
\end{proof}

\subsection{Randomized MDTs with Query Restrictions}
Similar to the deterministic case, when the height is restricted to $\bigO(\log n)$, we can define $\RMDT{\calC}$ for a circuit complexity class $\calC$. By using threshold gates to compute the probability bounds, we show that $\RMDT{\calC} = \calC$ if $\textit{mon-}\TC^0 \subseteq \textit{mon-}\calC$.
By using a carefully constructed normal form for randomized monotone decision trees 
we then show that $\RMDT{\AC^0} \subseteq \AC^0$.

\begin{theorem}
	For any circuit complexity class $\calC$ such that $\textit{mon-}\TC^0\subseteq \textit{mon-}\calC$ and closed under polynomial $\vee,\wedge,\neg$, we have $\RMDT{\calC} = \calC$.
\end{theorem}
\begin{proof}
	To show $\RMDT{\calC} \subseteq \calC$, let $f:\{0,1\}^n\to \{0,1\}$ be any function in $\RMDT{\calC}$.
	By definition, there is an RMDT called $T$ with height $h=\bigO(\log n)$. Let $\ell_1,\ell_2,\dots \ell_k$ be all the 1-labeled leaves; $r_1,r_2,\dots r_k$ be the number of random nodes from root to the corresponding leaf; and let $c_i:=\bigwedge f_{pi} \bigwedge \overline{f_{qi}} $ denote the characteristic function corresponding to $\ell_i$, where the $f_{pi}$'s are the monotone queries to be passed and $f_{qi}$'s to be failed in the root-$\ell_i$ path. Notice that since the queries have circuits in $\calC$, so do the $c_i$ functions. By the characterization in the proof of Theorem~\ref{thm:rmdt-main} (Equation~\eqref{eqn:1}), we have $f(x)=1$ iff $\sum_i 2^{h-r_i}c_i(x)\ge 2^{h-1}$.
	
	The RHS may be written as $\sum_i P_i(n) c_i(x)\ge Q(n)$, where the coefficients $P_i$'s and $Q$ are polynomial in $n$, since $h=\bigO(\log n)$. 
	
	We have $f(x)=1 \iff \sum_i P_i(n) c_i(x)\ge Q(n)$. As polynomial weighted threshold can be done in $\TC^0$ (and hence has a circuit in $\calC$) and so can be the $c_i$'s, we can construct a circuit in class $\calC$ for the task on the RHS, which clearly is the function $f$ on the LHS.

	We therefore have $\RMDT{\calC} \subseteq \calC$. Since $\MDT{\calC} \subseteq \RMDT{\calC}$, and we already proved that $\MDT{\calC}=\calC$ (Theorem~\ref{introthm:mdtc-c}), we get $\RMDT{\calC}=\MDT{\calC}=\calC$.
\end{proof}

We have a partial result for query restricted randomized MDTs in case of $\AC^0$:

\begin{theorem}
	\label{rmdt-ac0}
	$\RMDT{\AC^0} \subseteq \AC^0$.
\end{theorem}

Before describing the proof, we will first obtain normal forms for Randomized Monotone Decision Tree model. Let any Boolean function $f$ on $n$ variables be computed by a RMDT called $T$. We will modify $T$ into an RMDT $T'$ so that $T'$ also computes $f$, and additionally the following properties hold for $T'$.

\begin{itemize}
	\item \textbf{The number of random nodes is the same in all the paths from root of $T'$ to the leaves.} To achieve this, let $r$ be the maximum number of random nodes along any root to leaf path in $T$, and let $\ell$ be a leaf whose path to the root of $T$ contains $d(<r)$ many random nodes. Then we replace the leaf $\ell$ with a random node and make both its children leaves with label same as $\ell$. Now, both the newly added leaves can be seen to be at a distance of $d+1$ from the root of $T$. It is important to note that this modification does not change the probability of reaching a correct leaf. If the original leaf $\ell$ was reached in a computation with some probability, it is with the same probability that some child of $\ell$ will be reached (it does not matter which as both are labeled same as label of $\ell$). We perform this operation for all the leaves of $T$ that are at distance less than $h$ from the root of $T$, until no such leaves exist. Thus, in the final tree $T'$, all the paths contain the same number of random nodes (namely $r$).
	
	\item \textbf{$T'$ is a complete binary tree}. If $T$ already satisfies this property, we are done. Otherwise, suppose $\ell$ is some leaf of $T$ whose distance (i.e., number of queries) from the root is $d<h$, where $h$ is the height of $T$. Now consider the following transformation of $T$. We replace $\ell$ with the monotone query \textbf{1}, and label its left leaf arbitrarily and the right leaf with the label of $\ell$. The two newly introduced leaves can be seen to be at a distance of $d+1$ from the root of $T$. Also, $T$  still computes the same function as in any event that the leaf $\ell$ is reached upon computation on an input in the original $T$, the computation in the new $T$ ends at the newly introduced right leaf, whose label is the same as $\ell$. We keep performing this transformation for all the leaves that are at distance less than $h$ from the root, until finally all leaves are at a distance of $h$ from the root.
	
	Note that the above properties can both be made satisfied by $T'$ by making the same transformations; but first make the number of random nodes same in all paths, and only then make all paths the same length.
	
	\item \textbf{All the random nodes are at the top levels of $T'$.} By this, we mean that in any path from root to a leaf, the first few nodes are random, and the remaining are monotone nodes. For this, we first perform the above two transformations on $T$, after which say there are $r$ random nodes in every path and each path length is same as the height $h$ of $T$. 
	
	The idea to achieve this is: instead of making random queries in the process of computation by the monotone queries, we first fix a ``randomness'' and then proceed with the computation by the monotone queries. 
	
	We define $2^r$ many deterministic versions of $T$ as the set $\{T_s~|~\text{$s$ is a binary sequence of length $r$}\}$. A tree $T_s$ has the same structure as $T$ in terms of monotone nodes, leaves and their labels. A random node $v$ is replaced with the monotone function \textbf{0} if the $i^{th}$ bit of $s$ is 0, and with \textbf{1} otherwise, where $i$ is the number of random nodes in the path from $v$ to the root of $T$. 
	
	The new tree $T'$, whose height will be $r+h$ is constructed as follows: The top of $T'$ shall be a complete binary MRDT of height $r$ with all the nodes making random queries. This results in $2^r$ ``dangling ends''. For all binary sequences $s$, the dangling end \textit{addressed} by the binary sequence $s$ shall lead to the root of $T_s$. This completes the construction of $T'$. Immediately observe that $T'$ satisfies the desired property of all the random nodes being on the top. Now, to show that $T'$ is equivalent to $T$, we will argue that for any input $x$, the probability of reaching 1 in $T$ is equal to that of reaching 1 in $T'$. The former is equal to $\sum_i (1/2)^r.c_i(x)$, where $i$ runs over all the 1-labeled leaves $\ell_1,\ell_2,\dots$ of $T$ and $c_i$'s are their respective characteristic functions defined in Theorem 17. As there are $2^r$ leaves in $T'$ corresponding to each leaf $\ell_i$ in $T$, the latter probability is equal to $\sum_{i} \sum_s (1/2)^r.c^s_i(x)$. Here $c^s_i$ denotes the characteristic function of the leaf corresponding to $\ell_i$ in the tree $T_s$. To show the equality of these two probabilities, we will show that for any $i$, $c_i(x) = \sum_s c^s_i(x)$. As the trees $T_s$'s are same as $T$ with random nodes changed to \textbf{0} or \textbf{1}, observe that for any $s$, $c^s_i(x) \Rightarrow c_i(x)$. Now, suppose $c^s_i(x)=1$ for some $s$. In addition to $c_i(x)=1$, this also means that the characteristic product corresponding only to the originally random nodes in the root to $\ell_i$ path of $T_s$ is 1. This means that each of the literals in the product is 1, which translates to the fact that we can determine whether each of the originally random nodes of $T_s$ were \textbf{0} or \textbf{1} (since we know the root to $\ell_i$ path). But this uniquely determines an $s$ due to the definition of $T_s$. Thus at most one of $c^s_i(x)$ would be 1. Using this along with $\forall s, c^s_i(x) \Rightarrow c_i(x)$, we get the desired equation $c_i(x)=\sum_s c^s_i(x)$. Hence $T'$ computes $f$ too.
\end{itemize}

\begin{proof}[Proof of Theorem~\ref{rmdt-ac0}]
	Let a Boolean function $f$ on $n$ variables have an RMDT called $T$ of height $h=\bigO(\log n)$ with all its monotone queries having monotone $\AC^0$ circuits. From the three normal forms proved above, we may assume that the first $r$ levels of $T$ are random nodes, and the below queries all have monotone $\AC^0$ circuits. This is justified, since those normal forms do not increase the tree height beyond $\bigO(\log n)$, and the query functions still have monotone $\AC^0$ circuits.
	
	
	As there are no random queries beyond height $r$, each sub-tree rooted at a node just below a random node in $T$, resembles a (deterministic) MDT. Since we have established $\MDT{\AC^0} \subseteq \AC^0$ (Lemma~\ref{lem:mdtc-c}), each of these sub-trees computes an $\AC^0$ function, say $f_1,f_2,\dots f_{2^r}$. For an input $x$, recall that it means that probability of reaching a leaf with label $f(x)$ in $T$ is at least 2/3, which in other words means that at least 2/3$^{rd}$ of all leaves that are reachable with non-zero probability are labeled $f(x)$. Since the sub-trees are deterministic, there is a unique leaf reached in a sub-tree for a given $x$. Hence, at least 2/3$^{rd}$ of the sub-trees reach $f(x)$ labeled leaf (correspondingly $f_i(x)=f(x)$) on input $x$. Therefore, $f$ is essentially a majority function over $f_i$'s, with the added advantage that the majority bits are at least 2/3$^{rd}$ of total. Ajtai and Ben-Or \cite{AB84} proved the existence of an $\AC^0$ circuit that computes majority in such cases. As all the $2^r=\poly(n)$ many functions $f_i$'s are also in $\AC^0$, we can obtain an $\AC^0$ circuit for $f$. Thus, $f\in \AC^0$ and $\RMDT{\AC^0} \subseteq \AC^0$.
	
\end{proof}

\section{Discussion and Open Problems}

We explore the power of deterministic MDTs (adaptive and non-adaptive) and MDLs with most general monotone queries, and establish the relations between the corresponding complexity measures and alternation (Sections~\ref{sec:d-mdt} and~\ref{sec:d-namdt}). We also introduce NMDTs and RMDTs and understand their power (Sections~\ref{sec:ndt} and~\ref{sec:rdt}). Exploring the case of restricted queries, we show containments between various circuit complexity classes and the corresponding deterministic and randomized MDT classes (Section~\ref{sec:mdt-query-restrict}). By using a construction due to \cite{SW93}, we prove an upper bound on the number of negations in $\AC^0$ circuits, and justify its role in potentially solving $\MDT{\AC^0}=_?\AC^0$ (Sub-section~\ref{subsec:mdt-ac0}).

The question of $\MDT{\AC^0}=_?\AC^0$ is one of the main problems left unanswered by us. It is not even known whether the simple function like $f=\overline{x_1}x_2 \lor \overline{x_3}x_4 \lor \dots \lor \overline{x_{n-1}}x_n$ is in $\MDT{\AC^0}$, or even in $\RMDT{\AC^0}$. In this direction, we first note that the number of negations used in Theorem~\ref{introthm:ne-ac0} cannot be improved asymptotically, for if we can, then we can show that any function in $\TC^0$ can be computed by a $\TC^0$ circuit with $o(n^{\epsilon})$ negations (for any $0<\epsilon<1$), which contradicts Corollary~3.3(2) of \cite{SW93}. We note that if the negations bound in Theorem~\ref{thm:neg-ac0-lb} can be improved for some function, that shows the separation. If it cannot be improved, that means every function in $\AC^0$ can be computed using an $\AC^0$ circuit with ${\cal{O}}(n^{\epsilon})$ negation gates for arbitrarily small constant $\epsilon > 0$, which would result in an improvement of Theorem~\ref{thm:sublinear-ac0}.

We also comment here about an alternative way to define the classes $\MDT{\calC},\RMDT{\calC}$ and $\MDL{\calC}$, where we can restrict the queries by imposing that they be monotone but have possibly non-monotone circuits in $\calC$, rather than monotone circuits. This version results in potentially more powerful classes. With very similar proofs, we would still get $\MDT{\calC}=\RMDT{\calC}={\calC}$ when ${\calC} \supseteq \TC^0$. Further, Proposition~\ref{prop:equal-alt-decomp} implies that for any function $f \in \TC^0$ with uniform alternation, we get $f \in \MDTTT{\lceil{\log (\alt(f)+1)} \rceil}{\TC^0}$. The effect on our other results is unclear.

\bibliographystyle{plain}
\bibliography{ref}

\begin{thebibliography}{10}

\bibitem{AB84}
Miklos Ajtai and Michael Ben-Or.
\newblock A theorem on probabilistic constant depth computations.
\newblock In {\em Proceedings of the Sixteenth Annual ACM Symposium on Theory
  of Computing}, page 471–474, 1984.

\bibitem{Ant02}
Martin Anthony.
\newblock Decision lists.
\newblock Technical report, 2002.
\newblock CDAM Research Report LSE-CDAM-2005-23.

\bibitem{BN95}
Y.~{Ben-Asher} and I.~{Newman}.
\newblock Decision trees with and, or queries.
\newblock In {\em Proceedings of 10th Annual Conference on Structure in
  Complexity Theory.}, pages 74--81, 1995.

\bibitem{Ben83}
Michael Ben-Or.
\newblock Lower bounds for algebraic computation trees.
\newblock In {\em Proceedings of the 15th ACM Symposium on Theory of
  Computing}, pages 80--86, 1983.

\bibitem{BLW92}
Anders Bj\"{o}rner, L\'{a}szl\'{o} Lov\'{a}sz, and Andrew C.~C. Yao.
\newblock Linear decision trees: Volume estimates and topological bounds.
\newblock In {\em Proceedings of the 24th ACM Symposium on Theory of
  Computing}, pages 170--177, 1992.

\bibitem{BCOST15}
Eric Blais, Cl{\'e}ment~L. Canonne, Igor~C. Oliveira, Rocco~A. Servedio, and
  Li-Yang Tan.
\newblock {Learning Circuits with few Negations}.
\newblock In {\em Approximation, Randomization, and Combinatorial
  Optimization}, volume~40, pages 512--527, 2015.

\bibitem{Blu92}
Avrim Blum.
\newblock Rank-r decision trees are a subclass of r-decision lists.
\newblock {\em Information Processing Letters}, 42(4):183 -- 185, 1992.

\bibitem{Bsh96}
Nader~H. Bshouty.
\newblock A subexponential exact learning algorithm for dnf using equivalence
  queries.
\newblock In {\em Information Processing Letters}, pages 37--39, 1996.

\bibitem{BW02}
Harry Buhrman and Ronald de~Wolf.
\newblock Complexity measures and decision tree complexity: A survey.
\newblock {\em Theor. Comput. Sci.}, 288(1):21--43, 2002.

\bibitem{DL78}
David Dobkin and Richard~J. Lipton.
\newblock A lower bound of n on linear search programs for the knapsack
  problem.
\newblock {\em J. Comput. Syst. Sci}, pages 413--417, 1978.

\bibitem{Fis75}
Michael~J Fischer.
\newblock The complexity of negation-limited networks—a brief survey.
\newblock In {\em Automata Theory and Formal Languages, 2nd GI Conference
  Kaiserslautern}, pages 71--82, 1975.

\bibitem{POSSW11}
Parikshit. Gopalan, Ryan. O'Donnell, Rocco~A. Servedio, Amir. Shpilka, and
  Karl. Wimmer.
\newblock Testing fourier dimensionality and sparsity.
\newblock {\em SIAM Journal on Computing}, 40(4):1075--1100, 2011.

\bibitem{GLR01}
David Guijarro, Victor Lavin, and Vijay Raghavan.
\newblock Monotone term decision lists.
\newblock {\em Theoretical Computer Science}, 259(1):549 -- 575, 2001.

\bibitem{huang2019induced}
Hao Huang.
\newblock Induced subgraphs of hypercubes and a proof of the sensitivity
  conjecture.
\newblock {\em Annals of Mathematics}, 190(3):949--955, 2019.

\bibitem{Juk12}
Stasys Jukna.
\newblock {\em Boolean Function Complexity: Advances and Frontiers}.
\newblock Springer Berlin Heidelberg, Berlin, Heidelberg, 2012.

\bibitem{KM93}
Eyal Kushilevitz and Yishay Mansour.
\newblock Learning decision trees using the fourier spectrum.
\newblock {\em SIAM J. Comput.}, 22(6):1331--1348, 1993.

\bibitem{Markov58}
A.~A. Markov.
\newblock On the inversion complexity of a system of functions.
\newblock {\em J. ACM}, 5(4):331--334, October 1958.

\bibitem{SW93}
M.~Santha and C.~Wilson.
\newblock Limiting negations in constant depth circuits.
\newblock {\em SIAM Journal on Computing}, 22(2):294--302, 1993.

\bibitem{San19}
Swagato Sanyal.
\newblock Fourier sparsity and dimension.
\newblock {\em Theory of Computing}, 15(11):1--13, 2019.

\bibitem{Sip83}
Michael Sipser.
\newblock Borel sets and circuit complexity.
\newblock In {\em Proceedings of the Fifteenth Annual ACM Symposium on Theory
  of Computing}, STOC 83, pages 61--69, 1983.

\bibitem{Sni81}
Marc Snir.
\newblock Proving lower bounds for linear decision trees.
\newblock In {\em International Conference on Automata, Languages and
  Programming}, pages 305--315, 1981.

\bibitem{SY80}
J~Michael Steele and Andrew~C Yao.
\newblock Lower bounds for algebraic decision trees.
\newblock Technical report, Department of Computer Science, Stanford
  University, 1980.

\bibitem{YR80}
Andrew~C. Yao and Ronald~L. Rivest.
\newblock On the polyhedral decision problem.
\newblock {\em SIAM Journal on Computing}, 9(2):343--347, 1980.

\bibitem{SZ10}
Zhiqiang Zhang and Yaoyun Shi.
\newblock On the parity complexity measures of boolean functions.
\newblock {\em Theoretical Computer Science}, 411(26):2612 -- 2618, 2010.

\end{thebibliography}

\ifthenelse{\equal{\movetoappendix}{1}}{
        \appendix
        \section{Appendix}
        \includecollection{appendix}
} { }
\end{document}